\documentclass[a4paper,UKenglish,cleveref, autoref, thm-restate]{lipics-v2021}

\usepackage{float}
\usepackage{thm-restate}
\usepackage{algorithm}
\usepackage{algpseudocode}

\usepackage{amsthm}
\usepackage{verbatim}
\usepackage{etoolbox}

\newtoggle{showproofs}
\settoggle{showproofs}{true}   

\iftoggle{showproofs}{
}{
  \renewenvironment{proof}{\comment}{\endcomment}
}

\algrenewcommand{\algorithmicrequire}{\textbf{Input:}}
\algrenewcommand{\algorithmicensure}{\textbf{Output:}}

\newcommand{\Dist}{\mathrm{Dist}}
\newcommand{\Post}{\mathrm{Post}}
\newcommand{\Pred}{\mathrm{Pred}}
\newcommand{\Nat}{\mathbb{N}}
\newcommand{\Prob}{\mathbb{P}}

\def\alain#1{\textit{\textcolor{blue}{AF: #1}}}

\def\raphael#1{\textit{\textcolor{red}{RF: #1}}}

\def\lina#1{\textit{\textcolor{purple}{LY: #1}}}

\def\gaspard#1{\textit{\textcolor{brown}{GF: #1}}}

\title{Quantitative coverability for probabilistic well-structured transition systems} 

\titlerunning{Quantitative coverability for probabilistic WSTS} 

\author{Raphaël {Faure}}{Universit\'e Paris-Saclay,  CNRS, 
ENS Paris-Saclay, Laboratoire Méthodes Formelles, 
91190, Gif-sur-Yvette, France \and \url{https://home.lmf.cnrs.fr/RaphaelFaure} }{raphaelfaure@lmf.cnrs.fr}{https://orcid.org/0009-0001-1292-9090}{}

\author{Alain {Finkel}}{Universit\'e Paris-Saclay, CNRS, 
ENS Paris-Saclay, Laboratoire Méthodes Formelles,  
91190, Gif-sur-Yvette, France \and \url{https://ens-paris-saclay.fr/alain-finkel/} }{alain.finkel@ens-paris-saclay.fr}{https://orcid.org/0000-0003-2482-6141}{}

\author{Gaspard {Fougea}}{Universit\'e Paris-Saclay,  CNRS, 
ENS Paris-Saclay, Laboratoire Méthodes Formelles, 
91190, Gif-sur-Yvette, France \and \url{https://lmf.cnrs.fr/Perso/GaspardFougea} }{gaspardfougea@lmf.cnrs.fr}{https://orcid.org/0009-0004-8357-5340}{}

\author{Lina {Ye}}{Universit\'e Paris-Saclay, CNRS, 
ENS Paris-Saclay, CentraleSupélec, Laboratoire Méthodes Formelles,  
91190, Gif-sur-Yvette, France \and \url{https://home.lmf.cnrs.fr/LinaYe/} }{lina.ye@universite-paris-saclay.fr}{https://orcid.org/0000-0002-2217-4752}{}

\authorrunning{R. Faure et al.} 

\Copyright{} 

\category{} 

\relatedversion{}

\acknowledgements{Alain Finkel thanks Jean Goubault-Larrecq and Serge Haddad for initial discussions on probabilistic WSTS.}

\begin{document}

\nolinenumbers

\maketitle

\begin{abstract}
    Well-structured transition systems (WSTS) provide a classical framework for
the verification of infinite-state systems, but their probabilistic
extensions lack a unified treatment of quantitative coverability:
path-enumeration algorithms assume a finite branching degree, while
alternative approximation schemes defer some computations, such as
probabilities over a bounded horizon, to the model at hand. We introduce
probabilistic well-structured transition systems (pWSTS), Markov chains over
countable state sets whose underlying transition systems are WSTS, with no \emph{a
priori} assumption on the branching degree. This class encompasses any WSTS equipped 
with a Markov kernel, such as probabilistic vector addition systems (pVAS) 
and probabilistic lossy channel systems (pLCS). For an effective subclass, we solve
the approximate quantitative coverability problem over bounded horizons, and
over infinite horizons under decisiveness, requiring no probabilistic information beyond individual transition probabilities. We then identify a general source of decisiveness: every stochastically
monotone pWSTS is decisive with respect to every upward-closed set. We finally instantiate the framework on multi-type Galton--Watson processes,
a classical model of population dynamics whose offspring distributions may
have infinite support. Under mild assumptions on the reproduction laws, these processes 
are effective pWSTS, and they are stochastically monotone, hence decisive. Approximate 
quantitative coverability is therefore computable for them over both horizons, with a
proof that uses none of the traditional tools: neither generating functions
nor any case distinction between regimes.
\end{abstract}

\section{Introduction}
\label{sec:intro}

\textbf{General context.} 
Well-structured transition systems (WSTS)
\cite{Finkel87,FinkelSchnoebelen,AbdullaCeransJonssonTsay2000}
form a classical framework for the verification of infinite-state systems.
They combine a well-quasi-order with a monotone transition relation, yielding
finite representations of upward-closed sets and, under suitable effectiveness
assumptions, backward algorithms for coverability. Prominent examples include Petri nets/vector addition systems (VAS), lossy channel systems (LCS), and several classes of parameterized protocols. 

Many probabilistic infinite-state models whose underlying transition systems are WSTS have nevertheless been studied separately. Examples include probabilistic VAS (pVAS) and probabilistic LCS (pLCS).

Algorithms for approximating reachability probabilities do apply to these
models, but each line of work carries its own trade-off:
path-enumeration algorithms assume a finite branching degree, while alternative approximation schemes, which do handle infinitely
branching systems, defer some computations to the model at hand, such as the probabilities over a bounded horizon, the horizon being
the number of steps over which the process is observed (it may be infinite,
in which case runs of any length are considered).  
What is missing is an approximation scheme that does not depend on the
branching degree, and that queries the model for individual transition
probabilities rather than for bounded-horizon ones. Such a scheme would leave to each model only the verification of a fixed set
of effectiveness conditions, the approximation being then carried out
uniformly over the class. This motivates a unified probabilistic WSTS framework for the
probabilistic versions of these questions, coverability first among them.

\textbf{Related work and gap.}
Several approaches study probabilistic infinite-state systems and the
approximate quantitative reachability problem. To tackle this problem,
decisiveness emerged as a central property guaranteeing that, with probability
one, a run eventually reaches either a target or a state from which that target
is unreachable~\cite{AbdullaBenHendaMayr2007}.

Under suitable effectiveness assumptions, decisiveness supports approximation
algorithms for reachability probabilities, among which the classical ones, based on enumerating paths, typically rely on finite branching~\cite{IyerNarasimha1997,AbdullaBenHendaMayr2007}. Moreover,
\cite{bertrand_et_al:LIPIcs.ICALP.2016.101} generalizes path-enumeration to a
larger class of stochastic processes, allowing for continuous state set and
uncountable branching. Their approximation scheme, however, takes the
bounded-horizon probabilities as given, their computation being deferred to the
model at hand. This is not a trivial requirement: since the branching may be
infinite, path enumeration cannot be used to obtain these quantities, and the
model must supply them by other means, typically through sound abstractions
such as region graphs. Our framework restricts state sets to countable ones,
but asks nothing of the model beyond individual transition probabilities.

A complementary line of work exploits transience and truncation. Divergence
offers a treatment of transient systems by identifying regions that can be
safely truncated~\cite{FinkelHaddadYeDynamic2023}. The effectiveness results
there, however, are established for finitely branching systems.

More recently, Aspnes~\cite{AspnesSWSTS} took a step towards a unified
framework with \emph{stochastic well-structured transition systems} (SWSTS),
aimed at distributed computing models such as population protocols and
chemical reaction networks. An SWSTS is a WSTS equipped with a Markov kernel
and a weight function on states, in which every transition can be matched from
any larger state with probability at least inverse-polynomial in the weight of
the current state. When the system is \emph{closed}, that is when every
transition preserves the weight, the expected time before a run reaches a
given upward-closed set or a state from which it is unreachable is polynomial
in the initial weight. Without closedness the conclusion may fail, as a
non-closed chemical reaction network shows~\cite[Section~5.3.1]{AspnesSWSTS}.
More fundamentally, the weight function is an extra quantitative requirement,
induced neither by the well-quasi-order nor by the kernel. Moreover, WSTS models in which objects can be added or consumed generally do not satisfy the closed-case condition with their natural weight. 
This is the case for pLCS: with the channel
contents' length as natural weight, both sending and losing a message change
it.

These approaches do not provide a general WSTS-level framework yielding
effective approximation schemes for coverability probabilities where the
branching degree is left completely unspecified and where the only
probabilities the model has to supply are those of individual transitions. We
identify structural conditions under which an \(\varepsilon\)-approximation
scheme can be established for any upward-closed target set and any initial state. 

\textbf{Approximate Quantitative Coverability Problem.} Given an effectively presented WSTS with a Markov kernel, a finite basis \(B\), an
initial state \(x\), a horizon \(\ell\in\Nat\cup\{\infty\}\), and a precision
\(\varepsilon\), the problem is to compute the probability of reaching the upward closure of $B$ from $x$ within $\ell$ steps up to precision $\varepsilon$. We
consider both bounded and infinite horizons.\\

\textbf{Contributions.}
\begin{romanenumerate}
\item We define probabilistic WSTS (pWSTS), a model class encompassing both WSTS equipped with
a Markov kernel and stochastic processes whose underlying transition systems
are WSTS, with no a priori bound on the branching degree.  

\item For an effective subclass, we solve the approximate quantitative coverability problem over bounded horizons, and over infinite horizons under decisiveness, without assuming finite branching or bounded-horizon probabilities. Infinite branching is handled by successor enumeration guided by the remaining unexplored probability mass, computed solely from transition probabilities.

\item We show that multi-type Galton Watson processes, a classical model of
population dynamics with potentially unbounded offspring distributions, are
pWSTS. Furthermore, we establish its Pred-Basis effectiveness\footnote{Given a finite basis B of an upward-closed set, one can effectively compute a finite basis of its predecessor set.}, 
under mild assumptions on their
reproduction laws. 
\item We prove that stochastically monotone pWSTS are decisive with respect to every upward-closed set, and that multi-type Galton--Watson processes are stochastically monotone. To the best of our knowledge, this is the first pWSTS treatment of these processes, yielding a generic procedure for approximate quantitative coverability over both bounded and infinite horizons.
\end{romanenumerate}

\noindent
\textbf{Roadmap.}
Sections~\ref{sec:preli} and~\ref{sec:pWSTS} provide the necessary background and introduce the pWSTS framework along with its effectiveness assumptions. Section~\ref{sec:Path-Enumeration} presents the extended path-enumeration algorithms, while Section~\ref{sec:applications}
 details the application to multi-type Galton--Watson processes. Finally, Section~\ref{sec:ccl}
 concludes the paper, and a statement on the use of AI
follows.


\section{Preliminaries}
\label{sec:preli}
This section establishes the preliminary concepts and notation used throughout the paper. It is divided into two parts: the first reviews well-structured transition systems (WSTS), and the second recalls fundamental notions of Markov chains to formally define our probabilistic framework.

\subsection{WSTS}

For a non-empty set $S$, a \emph{quasi-order} (qo) $\leq$ is a binary relation on $S$ that is reflexive and transitive. $(S, \le)$ is a \emph{well-quasi-order} (wqo) if every infinite sequence $(x_i)_i \in S^\Nat$ contains two elements $x_i \le x_j$ with $i<j$. A transition relation $\to \subseteq S \times S$ on a qo $(S, \le)$ is \emph{strongly monotone} if for all $x,y,x'\in S$, $(x\leq x' \ \wedge\ x\to y)$ $\Longrightarrow (\exists y'\in S:\ x'\to y' \ \wedge\ y\leq y').$ 

\begin{definition}[Well-Structured Transition System \cite{Finkel87, FinkelSchnoebelen, AbdullaCeransJonssonTsay2000}]
\label{def:wsts}
A \emph{Well-Structured Transition System with strong monotonicity} (WSTS) is a tuple $\mathcal{W}=(S,\le,\to)$, where
\begin{itemize}
    \item $S$ is a set of states;
    \item $\le$ is a well-quasi-order;
    \item $\to \subseteq S\times S$ is strongly monotone.
\end{itemize}
\end{definition}

Several closely related, but not totally equivalent, definitions appear in the literature under the name of Well-Structured Transition System. They may differ in particular in the monotonicity requirement (see \cite{FinkelSchnoebelen} for a presentation of the different monotonicities). In this paper, we use the term WSTS in the strong monotonicity sense above.

Let $\mathcal{S}=(S,\leq,\rightarrow)$ be a WSTS. A \emph{run} is a sequence of transitions 
$s_0 \rightarrow s_1 \rightarrow \cdots$,
where each consecutive pair is related by $\rightarrow$. Runs are also called \emph{paths} in the literature, whence the standard
term \emph{path-enumeration algorithms}, which we keep. A finite run of
length $m$ from $s$ to $t$ is written $s \rightarrow^{m} t$, with the
convention that $s \rightarrow^{0} s$. We write $\rightarrow^{*}$ for the
reflexive transitive closure of $\rightarrow$. In the sequel, a run $s_0 \rightarrow s_1 \rightarrow \cdots$ is denoted by $s_0s_1...$ .

For a set $X \subseteq S$, we define the predecessor operators:
\begin{itemize}
  \item for $m\in \Nat$, $\Pred^m(X) = \{\, s \in S \mid \exists\, t \in X,\ s \to^m t \,\}$, in particular $\Pred^0(X) = X$ and $\Pred^1(X)$ is written $\Pred(X)$;
   \item for $\ell\in \Nat$, $\Pred^{\leq \ell}(X) = \bigcup_{0 \leq m \leq \ell} \Pred^m(X)$ and
  $\Pred^*(X) = \bigcup_{m \geq 0} \Pred^m(X)$. 
\end{itemize}
The successor operators $\Post^m(X)$, $\Post^{\le \ell}(X)$, and $\Post^{*}(X)$ are defined analogously. For $x \in S$,
we write $\Pred(x)$ for $\Pred(\{x\})$ and $\Post(x)$ for $\Post(\{x\})$, and likewise for the iterated operators. 

For $X\subseteq S$, we write $\uparrow X=\{y\in S \mid \exists x\in X,\ x\le y\}$ for its upward-closure, and $\uparrow x=\uparrow\{x\}$ for $x\in S$. A set $U\subseteq S$ is said to be \emph{upward-closed} if $U=\uparrow U$. By strong monotonicity of the transition relation, if $X$ is upward-closed, then $\Pred(X)$ is upward-closed, and therefore so are $\Pred^m(X)$ for all $m\in\mathbb{N}$ and $\Pred^*(X)$. A fundamental well-known property of well-quasi-orders~\cite{Higman1952} is that every upward-closed subset $U\subseteq S$ admits a finite basis, i.e., there exists a finite set $B\subseteq U$ such that $U=\uparrow B$.

\begin{lemma}\cite{AbdullaCeransJonssonTsay2000}
\label{lem:finite-depth}
Let $\mathcal S=(S,\le,\to)$ be a WSTS and let $U\subseteq S$ be upward-closed.
Then there exists $\ell\in\Nat$ such that 
$\Pred^*(U)=\Pred^{\le \ell}(U)$.

\end{lemma}

Furthermore, let us recall the notion of Pred-Basis effectiveness. 

\begin{definition}\cite{AbdullaCeransJonssonTsay2000}
\label{def:effective-wsts}
A WSTS
$(S,\le,\to)$ is \emph{Pred-Basis effective} if
$\le$ is decidable, and, 
for every finite set $B$, a finite basis of $\mathrm{Pred}(\uparrow B)$ can be computed. 
\end{definition}

For a Pred-Basis Effective WSTS and a finite set \(B\),
Lemma~\ref{lem:finite-depth} makes it possible to compute a finite basis of
\(\Pred^*(\uparrow B)\) using the backward coverability
algorithm~\cite{Abdulla_1996,FinkelSchnoebelen}.

\subsection{Markov Chains}

This section formally defines the probabilistic framework underlying our approach. For any set $S$, we denote by $2^S$ its power set, equipped as a $\sigma$-algebra.

\begin{definition}
    Let $S$ be a countable set. We denote by $\Dist(S)$ the set of discrete probability distributions over $2^S$.
\end{definition}
The support of a distribution $\mu \in \Dist(S)$ is defined as
$\mathrm{supp}(\mu) = \{ t \in S \mid \mu(\{t\}) > 0 \}$. To lighten notation, we write $\mu(t)$ instead of $\mu(\{t\})$ for all $t \in S$ and $\mu \in \Dist(S)$.

\begin{definition}
Let $S$ be a countable set. A \emph{Markov kernel on $S$} is a mapping $P: S \to \Dist(S)$.
\end{definition}

Let $(S, P)$ be a countable set equipped with P a Markov kernel. In this paper, we let $(X_n)_{n\ge 0}$ denote the Markov chain on $S$
induced by the kernel $P$, i.e., with transition probabilities
$\mathbb P(X_{n+1}=t\mid X_n=s)=P(s)(t)$. Moreover, for any initial state $x \in S$, the Ionescu--Tulcea Theorem~\cite{IonescuTulcea1949,Kallenberg2021} guarantees the existence of a unique probability measure $\mathbb{P}_x$ on the space $\Omega = S^{\mathbb{N}}$ of infinite runs. In this setting, the initial state is $X_0 = x$, each $X_k \colon \Omega \to S$ denotes the $k$-th coordinate projection (i.e., $X_k((x_i)_{i \geq 0}) = x_k$), and the process is equipped with the natural filtration $(\mathcal{F}_n)_{n \geq 0}$ defined by $\mathcal{F}_n = \sigma(X_0, \dots, X_n)$, and $\mathcal{F} =
\sigma(X_k: k \geq 0)$ the product $\sigma$-algebra on $\Omega$.

We recall the definitions of stopping and hitting times: 

\begin{definition}
\label{def:stopping-time}
A random variable $\tau \colon \Omega \to
\mathbb{N} \cup \{\infty\}$ is a \emph{stopping time} with respect to $(\mathcal{F}_n)_{n \geq 0}$ if
for every $n \geq 0$, $\{\tau \leq n\} \in \mathcal{F}_n$ .
\end{definition}

\begin{definition}
\label{def:hitting-return}
For a subset $A \subseteq S$, the \emph{hitting
time} $\tau_A$ of A is defined by
$
  \tau_A = \inf\{n \geq 0: X_n \in A\},
$
with the convention $\inf \emptyset = \infty$. 
\end{definition}

It is a stopping time with respect to $(\mathcal{F}_n)_{n \geq 0}$:
for every $n \geq 0$, the events $\{\tau_A \leq n\} = \bigcup_{k \leq n}
\{X_k \in A\}$ belong to $\mathcal{F}_n$.  Decisive Markov chains were introduced by Abdulla et al.~\cite{AbdullaBenHendaMayr2007}, and further studied in~\cite{FinkelHaddadYeDynamic2023,BarbotBouyerHaddad2024, bertrand_et_al:LIPIcs.ICALP.2016.101}.  
In this paper, for any $A \subseteq S$, we denote by $\widetilde A$ the set $(S \setminus \Pred^*(A))$ of states from which $A$ is no longer reachable.

\begin{definition}[Decisiveness]
The Markov chain induced by $(S,P)$ is \emph{decisive with respect to $A$} if,
from every initial state, almost every run eventually reaches either $A$ or $\widetilde A$.
\end{definition}

Let us remark that the Markov chain induced by $(S, P)$ is decisive w.r.t. $A$ if and only if $\Prob_x(\tau_{A\cup\widetilde A} < \infty) = 1$, for all $x\in S$. 

\section{Probabilistic WSTS and Stochastically Monotone pWSTS}
\label{sec:pWSTS}
From now on, all state sets are assumed to be countable. In this paper, we use the term \emph{procedure} for a computational scheme without guaranteed termination, \emph{semi-algorithm} when termination is ensured under specific additional conditions, and \emph{algorithm} when it is guaranteed for any input.
\subsection{Probabilistic WSTS}

Probabilistic models built over WSTS have been studied model by model:
stochastic Petri nets \cite{Molloy82}, pVASS
\cite{AbdullaBenHendaMayr2007,BenHenda08,AbdullaBenHendaMayrSandbergATVA2006,AbdullaBenHendaMayrSandbergQEST2006},
probabilistic lossy channel systems (pLCS)
\cite{IyerNarasimha1997,Schnoebelen2004,ABDULLA2005141,BenHenda08,AbdullaBenHendaMayrSandbergATVA2006,AbdullaBenHendaMayrSandbergQEST2006}.
We define below the class they all belong to, asking only that the
transition system underlying a Markov chain be a WSTS, without the
quantitative conditions imposed by SWSTS \cite{AspnesSWSTS}. We recall that $x \in S$ is a deadlock if $\Post(x) = \emptyset$.

\begin{definition}[Probabilistic WSTS]\label{def:pWSTS}
A \emph{probabilistic well-structured transition system (pWSTS)} is a pair $(\mathcal{S}, P)$, where $\mathcal{S} = (S, \leq, \to)$ is a WSTS and $P$ is a Markov kernel on $S$ compatible with the transition relation: \begin{itemize}
    \item for any non-deadlock state $x \in S$ and all $y \in S$, $P(x)(y) > 0 \Longleftrightarrow x \to y$;
    \item whereas for any deadlock state $x \in S$ , $P(x)(x) = 1$. 
\end{itemize} 
\end{definition}

In the literature, deadlocks are typically handled in two ways: either by
adding a self-loop in the transition system, or by making them absorbing
states with probability $1$ of remaining there. The first option is not
available here: adding $x \to x$ at a deadlock $x$ requires, by
monotonicity, that every $x' \geq x$ have a successor above $x$, which
nothing guarantees. We therefore model deadlocks as absorbing states, which
leaves the underlying WSTS untouched. Note that for any $x \in S$ such that $x$ is not a deadlock state, $\Post(x) = \mathrm{supp}(P(x))$. Before proceeding, we present two simple examples.





\begin{example}
\label{ex:one-counter}
Consider the system $S_1 = (\Nat,\leq,\to)$ whose state set is $\Nat$ with the standard order, having the following transition relation:
\begin{itemize}
    \item the increment step: for all $n\in \mathbb{N}$, $n \to n+1$;
    \item the decrement step: for all $n\in \mathbb{N}_+$, $n \to n-1$.
\end{itemize}
For $n\in\Nat_+$, we assign a probability of $1/2$ to each possible transition from $n$. At $0$, a probability of $1$ is assigned to the increment transition. This system is a pWSTS.

\end{example}
\begin{example}
\label{ex:infinite-jumps}
Consider the system $S_2 = (\Nat, \leq, \to)$ whose state set is $\Nat$
with the standard order, having the following transition relation:
\begin{itemize}
    \item the jump step: for all $n \in \Nat$ and all $k \in \Nat_+$,
    $n \to n+k$;
    \item the decrement step: for all $n \in \Nat_+$, $n \to n-1$.
\end{itemize}
For $n \in \Nat_+$, we assign the probability $1/2$ to the decrement
transition and $2^{-(k+1)}$ to each jump of size $k$. At $0$, where no
decrement is enabled, we assign $2^{-k}$ to each jump of size $k \in \Nat_+$. In both
cases the probabilities sum to one, $\le$ is a wqo and the transition relation is monotone,
so $S_2$ is a pWSTS. Its branching degree, however, is infinite: for $n \in \Nat_+$, $\Post(n)$
is the whole of $\{n-1\} \cup \{n+k \mid k \in \Nat_+\}$.
\end{example}

Let us recall the definition of pVAS and pLCS.

\begin{definition}[VAS]\label{def:VAS}
    Let $d \in \Nat_+$, a $d$-VAS (or a VAS of dimension $d$) is a finite set $V\subseteq \mathbb {Z} ^{d}$.
\end{definition}

Let $d \in \Nat_+$ and let $V\subseteq \mathbb {Z} ^{d}$ be a d-VAS. For $u \in \mathbb N^d$ and $v\in V$, we write $u \xrightarrow{v} u+v$ whenever $u+v\in\mathbb N^d$.

Observe that system $S_1$ is a $1$-VAS whereas $S_2$ is equivalent to an $\omega$-Petri net ($\omega$-Petri nets \cite{GeeraertsHPR15} are an infinite branching extension of Petri nets).

Moreover, let $w: V \to \mathbb{N}^+$ be a weight function. We define
transition probabilities as follows: for every configuration $u \in
\mathbb{N}^d$, let $\mathrm{En}(u) = \{v \in V: u + v \in \mathbb{N}^d\}$
denote its set of enabled actions. If $\mathrm{En}(u) \neq \emptyset$, then
for $v \in \mathrm{En}(u)$,
$$
  P(u)(u + v) = \frac{w(v)}{\sum_{v' \in \mathrm{En}(u)} w(v')},
$$
and $P(u)(x) = 0$ for every other state $x$. If $\mathrm{En}(u) =
\emptyset$, then $u$ is a deadlock and we set $P(u)(u) = 1$. Thus, $P$ defines a Markov kernel, and a d-VAS equipped with it is a pWSTS, since a d-VAS is a WSTS.

\begin{definition}[pLCS]
\label{def:pLCS}
A pLCS is a tuple $\mathcal{S} = (Q, \Sigma, \Delta, W, p_\ell)$ where:
\begin{itemize}
    \item $Q$ is a finite set of control states;
    \item $\Sigma$ is a finite alphabet of messages; 
    \item $\Delta \subseteq Q \times (\{!a, ?a \mid a \in \Sigma\} \cup \{\tau\}) \times Q$ is a finite set of transition rules; 
    \item $W: \Delta \to \mathbb{N}_+$ assigns a positive integer weight to each transition rule;
    \item $p_\ell \in (0,1)\cap \mathbb{Q}$ is a message loss probability.
\end{itemize}
\end{definition}

The configuration is a pair $(q,w)\in Q\times\Sigma^*$, where $w$ represents the contents of the FIFO channel. The state set is therefore $S = Q \times \Sigma^*$. We equip $\Sigma^*$ with the subword ordering relation $\sqsubseteq$, which is extended to configurations by $(q,w)\sqsubseteq(q',w') \iff q=q'$ and $w\sqsubseteq w'$. By Higman's Lemma, this ordering is a well-quasi-order on $S$. A transition is defined as, for $(q,w),(q',w') \in Q\times\Sigma^*$, $(q,w) \to (q',w')$ if and only if there exists $w'' \in \Sigma^*$ such that $(q,w) \to_{loss} (q,w'')$ and there exists $t \in (\{!a, ?a \mid a \in \Sigma\} \cup \{\tau\} )$ such that: $(q,t,q') \in \Delta$ and $(q,w'') \xrightarrow{t} (q',w')$. Where: \begin{itemize}
    \item $\to_{loss}$: for $(q,w),(q',w') \in Q\times\Sigma^*$, $(q,w) \to_{loss} (q',w')$, if and only if: $q = q'$ and $w' \sqsubseteq w$. This transition rule represents the loss of messages in the FIFO channel. Following \cite{ABDULLA2005141}, each message can be lost independently with probability $p_l$.
    \item for $q,q' \in Q$ and $a \in \Sigma$, $q \xrightarrow{!\,a} q'$ (\emph{emission}): always enabled, it appends the word $a$ to the tail of the channel, and move to control state $q'$;
  \item for $q,q' \in Q$, $a \in \Sigma$ and $w \in \Sigma^*$, $q \xrightarrow{?\,a} q'$ (\emph{reception}): enabled iff $w$
  starts with $a$; it removes the head of the channel and move to $q'$;
  \item for $q,q' \in Q$, $q \xrightarrow{\tau} q'$ (\emph{internal action}): always enabled;
  it moves to $q'$ without modifying the channels.
\end{itemize}

We assume that pLCS are deadlock-free, a property easily achieved by adding a self-loop $\tau$-transition to each control state. As shown in \cite{ABDULLA2005141}, pLCS can be extended with transition probabilities using the function $W$ and the parameter $p_\ell$. This probabilistic extension yields a Markov chain, which can formally be viewed as a pWSTS since the underlying lossy channel system with \emph{loss-action} semantics is already a WSTS. Note that the pLCS framework can also accommodate multiple channels.

\begin{proposition}
    Both pVAS and pLCS are pWSTS.
\end{proposition}

For models such as pVAS or pLCS, one starts from a WSTS and equips it with
probabilities. The reverse construction is also possible: from a Markov
chain on $S$, set $x \to y$ iff $P(x)(y) > 0$. Although the induced transition system is not necessarily a WSTS, when it is, the model falls within our framework, like the Galton--Watson processes of Section~\ref{sec:applications}.

\subsection{Stochastically Monotone pWSTS}
Recall that for $x \in S$ and $A \subseteq S$, $P(x)(A) = \sum_{y \in A} P(x)(y)$. Building on this, we use the concept of stochastic monotonicity, first introduced in~\cite{Daley1968StochasticallyMM} and later extended to partial orders in~\cite{KamaeKrengelOBrien1977}.

\begin{definition}
\label{def:stochastic-monotonicity}
A pWSTS $(\mathcal{S}, P)$ is \emph{stochastically monotone} if for all
$x, x' \in S$ with $x \leq x'$ and every upward-closed set $U \subseteq S$,
$$
  P(x)(U) \;\leq\; P(x')(U).
$$
\end{definition}

\begin{remark}
  We can define a quasi-order $\le_{st}$ over $\Dist(S)$ such that $\mu \le_{st} \nu$ if and only if for every upward-closed set $U \subseteq S$, $\mu(U) \;\leq\; \nu(U)$. For a stochastically monotone pWSTS, the relation $\le_{st}$ is a wqo over the image $P(S)$.
\end{remark}

A prominent example of a stochastically monotone pWSTS is the Galton--Watson process defined in Section~\ref{sec:applications}. In this subsection, we show that any stochastically monotone pWSTS is decisive with respect to any upward-closed set.

\begin{lemma}
\label{lem:monotone-reachability}
Let $(\mathcal{S}, P)$ be a stochastically monotone pWSTS and let
$U \subseteq S$ be upward-closed. Then for every $k \in \Nat$, the map
$f_k: x \mapsto \mathbb{P}_x(\tau_U \leq k)$ is non-decreasing.
\end{lemma}

\begin{proof}
We proceed by induction on $k$. For $k = 0$ we have
$f_0 = \mathbf{1}_U$, which is non-decreasing since $U$ is upward-closed.

Assume $f_k$ is non-decreasing and let $x \leq x'$. Recall
$$
  f_{k+1}(x) \;=\;
  \begin{cases}
    1 & \text{if } x \in U, \\[2pt]
    \displaystyle\sum_{y \in S} P(x)(y)\, f_k(y) & \text{otherwise.}
  \end{cases}
$$
If $x \in U$ then $x' \in U$ by upward-closedness, and
$f_{k+1}(x) = f_{k+1}(x') = 1$. If $x \notin U$ and $x' \in U$ then
$f_{k+1}(x') = 1 \geq f_{k+1}(x)$. There remains the case
$x, x' \notin U$, for which it suffices to prove
$\sum_y P(x)(y) f_k(y) \leq \sum_y P(x')(y) f_k(y)$.

Note that $P(x)(y) = \mathbb{P}_x(X_1 = y)$. For every $y \in S$, since
$f_k(y) \in [0,1]$, we have
$f_k(y) = \int_0^1 \mathbf{1}_{\{f_k(y) \geq t\}}\,\mathrm{d}t$, $\mathbf{1}_{\{f_k(y) \geq t\}}$ being equal to $1$ for $t \in [0, f_k(y)]$ and $0$ otherwise. All terms
being non-negative, the Fubini--Tonelli theorem gives
$$
  \sum_{y \in S} P(x)(y)\, f_k(y)
  \;=\; \mathbb{E}_x\big[f_k(X_1)\big]
  \;=\; \mathbb{E}_x\Big[\int_0^1 \mathbf{1}_{\{f_k(X_1) \geq t\}}\,\mathrm{d}t\Big]
  \;=\; \int_0^1 \mathbb{P}_x\big(f_k(X_1) \geq t\big)\, \mathrm{d}t .
$$
By the induction hypothesis $f_k$ is non-decreasing, so the level set
$L_t = \{\, y \in S \mid f_k(y) \geq t \,\}$ is upward-closed for every
$t \in [0,1]$, and $\mathbb{P}_x\big(f_k(X_1) \geq t\big) = P(x)(L_t)$.
Definition~\ref{def:stochastic-monotonicity} therefore yields
$P(x)(L_t) \leq P(x')(L_t)$ for every such $t$, and integrating over
$t \in [0,1]$ gives the required inequality.
\end{proof}

Let us recall the notion of bounded coarseness, as introduced by Abdulla et al.~\cite{AbdullaBenHendaMayrSandbergATVA2006}. A Markov chain
$(S, P)$ is \emph{boundedly coarse with parameter
$(\beta, K) \in (0,1] \times \Nat$ towards a target $A$} if for all
$x \in S$, either $x \notin \Pred^*(A)$ or
$\mathbb{P}_x(\tau_A \leq K) \geq \beta$. In what follows, we say that a pWSTS is decisive, boundedly coarse, and so on, when the Markov chain it induces is.

\begin{proposition}
\label{prop:boundedly-coarse-pWSTS}
Let $(\mathcal{S}, P)$ be a stochastically monotone pWSTS and let $U \subseteq S$ be upward-closed. Then $(\mathcal{S}, P)$ is boundedly coarse towards $U$.
\end{proposition}

\begin{proof}
If $U = \emptyset$, the result is immediate since $\Pred^*(U) = \emptyset$.
Otherwise, since $\leq$ is a wqo and $\Pred^*(U)$ is upward-closed, $\Pred^*(U)$ admits a finite basis
$B' = \{b'_1, \dots, b'_r\}$, with $r \geq 1$ because $U \neq \emptyset$.
For each $b'_j$, there exists a finite run $\pi_j$ from $b'_j$ to $U$. Denote by $k_j$ its length and by $\delta_j$ its probability. We have $k_j \ge 0$, $\delta_j > 0$, and $\mathbb{P}_{b'_j}(\tau_U \leq k_j) \geq \delta_j$. Set
$$
  K \;=\; \max\Big(1, \max_{1 \leq j \leq r} k_j\Big)
  \qquad \text{and} \qquad
  \beta \;=\; \min_{1 \leq j \leq r} \delta_j \;>\; 0 ,
$$
which are both well defined since $r$ is finite.

Let $y \in S$. If $y \notin \Pred^*(U)$, the condition of bounded coarseness is trivially satisfied. Assume now that $y \in \Pred^*(U)$. There is some $j$ such that $y \geq b'_j$. Since
$K \geq k_j$ and $\{\tau_U \leq k_j\} \subseteq \{\tau_U \leq K\}$,
Lemma~\ref{lem:monotone-reachability} gives
$$
  \mathbb{P}_y(\tau_U \leq K)
  \;\geq\; \mathbb{P}_y(\tau_U \leq k_j)
  \;\geq\; \mathbb{P}_{b'_j}(\tau_U \leq k_j)
  \;\geq\; \delta_j \;\geq\; \beta .
$$
Thus, $(S, P)$ is boundedly coarse with parameters $(\beta, K)$ towards $U$. 
\end{proof}

Recall that a boundedly coarse Markov chain is also globally coarse \cite{AbdullaBenHendaMayr2007}, meaning there exists a parameter $\delta > 0$ such that for every state $x \in S$, either $x \notin \Pred^*(U)$ or $\mathbb{P}_x(\tau_U < \infty) \geq \delta$. Furthermore, according to \cite{AbdullaBenHendaMayr2007}, any Markov chain that is globally coarse with respect to a set is decisive with respect to that same set. This immediately yields the following corollary.

\begin{corollary}
\label{cor:pWSTS-decisive}
A stochastically monotone pWSTS is decisive with respect to every upward-closed set.
\end{corollary}

\subsection{Effective pWSTS}


Before defining effectiveness for pWSTS, let us recall a few standard notions of computability. First, a real number $x$ is \emph{computable} if there exists an algorithm which, given a positive rational precision $\varepsilon$, returns rational numbers $a,b$ such that $a \le x \le b$ and $b-a\le \varepsilon$. In what follows, we further say that $P$ is \emph{computable} if there is an
algorithm which, given $x, y \in S$ and a rational precision
$\varepsilon > 0$, returns rationals $a, b$ with $a \leq P(x)(y) \leq b$ and
$b - a \leq \varepsilon$. Second, a mapping $\chi : S \to 2^S$, such as $\Post$, is \emph{effectively
enumerable} if there is an algorithm $A$ which, on input $x \in S$ and an
index $m \in \Nat_+$, either returns an element of $\chi(x)$ or does not terminate, and such
that: (i) every element of $\chi(x)$ is
returned for at least one index $m$; and (ii) $A$ may fail to terminate on
an index $m$ only if the elements returned on the indices $1, \dots, m-1$
already exhaust $\chi(x)$.



The definition of an \emph{effective Markov chain} of Abdulla et
al.~\cite{AbdullaBenHendaMayr2007} requires $\Post(x)$ to be computable
explicitly for each state $x$, which restricts their framework to finitely
branching models. We only require $\Post$ to be effectively enumerable and $P$ to be computable,
allowing $\Post(x)$ to be infinite: rather than
exploring all successors, we enumerate them while accumulating their
probabilities, which bounds the unexplored part of the distribution and
provides the stopping criterion.

For a given target set $F$, their approach also requires deciding membership in $\Pred^*(F)$, as well as a finite representation of $F$ with decidable membership. We replace both requirements by leveraging a wqo and the Pred-Basis assumption, restricting targets to upward-closed sets, which is a natural setting for the coverability problem.
\begin{definition}[Effective pWSTS]
A pWSTS $(\mathcal{S}, P)$ is effective if:
\begin{itemize}
    \item $\to$ is decidable;
 \item $=$ is decidable;
 \item $\mathcal{S}$ is a Pred-Basis effective WSTS;
 \item $\Post$ is effectively enumerable;
 \item $P$ is computable.
\end{itemize}
\end{definition}

An effective pWSTS can thus be finitely represented by a tuple of six
algorithms $(\to, =, \le, \mathrm{Pred\text{-}Basis}, \mathrm{succ}, P)$,
the decidability of $\le$ being part of Pred-Basis effectiveness.

Since $=$ is decidable and $\Post$ is effectively enumerable, we can restrict ourselves to a repetition-free effective enumeration, which we assume hereafter. Importantly, our framework is not restricted to rational probabilities but accommodates computable real probabilities. This broader scope allows us to consider models such as the Galton--Watson process (see Section~\ref{sec:applications}) with Poisson offspring distributions.

\begin{example}
\label{ex:infinite-jumps-effective}
The system $S_2$ of Example~\ref{ex:infinite-jumps} is an effective pWSTS.
Indeed, $\leq$ and $=$ are decidable on $\Nat$, and so is $\to$, since
$n \to m$ holds exactly when $m > n$, or $m = n-1$ with $n \in \Nat_+$. The
transition probabilities are rationals given by a formula, so
$P$ is computable. The map $\mathrm{succ}$ defined by
$\mathrm{succ}(n,1) = n-1$ and $\mathrm{succ}(n,m) = n+m-1$ for
$n \in \Nat_+$, and by $\mathrm{succ}(0,m) = m$, is computable and
enumerates $\Post(n)$ without repetition, so $\Post$ is effectively
enumerable. Finally, for any $b \in \Nat$, every state can reach
$\uparrow b$ in one jump, so $\Pred(\uparrow b) = \Nat$ and
$\{0\}$ is a computable finite basis of it: $S_2$ is
Pred-Basis effective.

This example shows that infinite branching is, by itself, no obstacle to
effectiveness.
\end{example}



Furthermore, it is well known that VAS and LCS are Pred-Basis effective WSTS with a finite and computable $\Post$. Moreover,  once
equipped with probabilities, $P$ is computable as sums or quotients of rationals.

\begin{proposition}
\label{prop:pvas-plcs-effective}
pVAS and pLCS are effective pWSTS.
\end{proposition}

Building on this definition, we state two propositions and one lemma that illustrate how the computability of $P$ can be combined with the effective enumerability of $\Post$. The first proposition shows that deadlock states can be identified. The second shows that, when transition probabilities are known exactly, $\Post$ can be effectively enumerated with a decidable termination test. The lemma then gives the practical counterpart: even without exact values, the enumeration can be stopped once a prescribed fraction of the probability mass has been covered. In what follows, and in the algorithms in particular, we represent the effective enumerability of $\Post$ by a computable function $\mathrm{succ}: S \times \Nat_+ \to S$. 

\begin{proposition}
\label{lem:decide-deadlock}
Let $(\mathcal{S}, P)$ be an effective pWSTS. It is decidable whether a given state $x \in S$ is a deadlock.
\end{proposition}
\begin{proof}
By definition, a state $x$ is a deadlock if $\Post(x) = \emptyset$. Equivalently, within the pWSTS setting, $x$ is a deadlock state if and only if $x \not\to x$ and $P(x)(x) = 1$. 

Since the pWSTS is effective, $\to$ is decidable. If $x \not\to x$, since $P(x)(x)$ must be either $0$ or $1$, it suffices to compute a lower bound for $P(x)(x)$ to a precision strictly better than $1/2$. If it returns a lower bound greater than $0$, we can deduce that $P(x)(x) = 1$, confirming $x$ is a deadlock. Otherwise, it is not.
\end{proof}

\begin{proposition}
\label{lem:exact-termination}
Let $(\mathcal{S}, P)$ be an effective pWSTS. Assume that for any $x,y \in S$, $P(x)(y) \in \mathbb{Q}$ and is computed exactly. Then, the termination of $\mathrm{succ}$ is decidable.
\end{proposition}
\begin{proof}
Let $x\in S$ and $m \in \mathbb{N}_+$. One first checks if $x$ is a deadlock state, which is decidable by Proposition~\ref{lem:decide-deadlock}. If $x$ is a deadlock, $\Post(x)$ is empty, so it does not terminate. If it is not, to decide whether $m > |\Post(x)|$, one enumerates the successors $z_1, z_2, \dots$ of $x$ and accumulates the exact sum $\sum_{i \leq k} P(x)(z_i)$. Since $\Post(x)$ is exactly the support of $P(x)$, the total mass over $\Post(x)$ is exactly $1$. If the accumulated sum reaches $1$ (which is decidable using exact values) at some step $p < m$, then $\Post(x)$ has been fully enumerated, which implies $m > |\Post(x)|$. Otherwise, if $\sum_{i \leq m-1} P(x)(z_i) < 1$, elements remain to be enumerated. 
\end{proof}

\begin{lemma}
\label{lem:approx-termination}
Let $(\mathcal{S}, P)$ be an effective pWSTS. Let $x \in S$ be a non-deadlock state and $\gamma \in \mathbb{Q}\cap(0,1)$. Enumerating elements of $\Post(x)$ until the accumulated probability mass
strictly exceeds $1 - \gamma$ terminates for a certain incremental and adaptive
precision on the transition probabilities.
\end{lemma}
\begin{proof}
Since $x$ is not a deadlock, $\Post(x) = \mathrm{supp}(P(x))$ and the total
mass $\sum_{z \in \Post(x)} P(x)(z)$ equals $1$. The probabilities being
computable reals, they are handled through rational bounds, and exact
equality tests such as $\sum_i P(x)(z_i) = 1$ may be undecidable. The
enumeration therefore accumulates rational lower bounds, the $i$-th term
computed to precision at most $\gamma\,2^{-(i+2)}$, so the total
approximation loss is bounded by $\sum_{i \geq 1} \gamma\,2^{-(i+2)} =
\gamma/4$. The computed partial sums thus increase to a value at least
$1 - \gamma/4$. Since $1 - \gamma/4 > 1 - \gamma$, they exceed $1 - \gamma$
at some finite index. Both the computed sum and the threshold being
rational, the comparison is decidable, and the loop halts there.
\end{proof}

Within this framework, one can effectively compute a run from any $x \in \Pred^*(\uparrow B)$ to $\uparrow B$. 

\begin{restatable}{proposition}{witnessRun}
\label{prop:witness-run}
Let $(\mathcal{S}, P)$ be an effective pWSTS. One can compute, for any finite set $B$ and any initial state $x \in \Pred^*({\uparrow}B)$, a shortest run from $x$ to the target ${\uparrow}B$, along with its probability.
\end{restatable}

Denote by $U$ the set ${\uparrow}B$. To give an instance of such a forward algorithm, 
we propose an approach that leverages the
structures already computed by the backward algorithm for deciding
membership in $U$ and $\Pred^*(U)$. Specifically, this algorithm uses the
intermediate bases to guide the path construction: the backward algorithm
computes a finite basis of $\Pred^{\leq k}(U)$ for each $k$, stabilizing at
some finite depth $\ell$ by Lemma~\ref{lem:finite-depth}. These bases can naturally direct the forward exploration: at each step, one enumerates the successors of the current state and simply selects the first element which is a step closer to $U$ with respect to these intermediate bases. Note that a blind interleaved search would not work: with $\Post$ only effectively enumerable, an out-of-range call to $\mathrm{succ}$ does not terminate, whereas the intermediate bases guarantee that a useful successor exists and is reached before the enumeration runs out. 

Formally, we introduce a membership function $\mathrm{In}_U$, computed in a
preprocessing step by the backward algorithm such that for $x \in S$ and $k \in
\mathbb{N}\cup\{-1\}$,
$$
  \mathrm{In}_U(x, k) =
  \begin{cases}
    1 & \text{if }k = -1 \text{ and }x \in \Pred^*(U), \\
    1 & \text{if }k\in \Nat \text{ and }x \in \Pred^{\leq k}(U), \\
    0 & \text{otherwise.}
  \end{cases}
$$
with $\mathrm{In}_U(x, 0) = 1$ iff $x \in U$. Since $U$ is upward-closed and finitely generated by $B$, and the pWSTS is effective, each $\Pred^{\le k}(U)$ admits a computable finite basis. Thus, $\mathrm{In}_U$ is computable by comparison with finitely many elements, as there exists a finite rank $\ell$ such that for all $k \ge \ell$, the bases of $\Pred^{\le \ell}(U)$ and $\Pred^{\le k}(U)$ are equal by Lemma~\ref{lem:finite-depth}. For any initial state $x \in S$, if $x \notin \Pred^*(U)$, Algorithm~\ref{alg:witness-path} outputs $\emptyset$ and $0$. Otherwise, without any additional assumptions, it outputs a minimum-length run $y_0 \cdots y_k$ such that $y_0 = x$, $y_i \notin U$ for $i < k$, and $y_k \in U$, with a strictly positive probability $\prod_{i=0}^{k-1} P(y_i)(y_{i+1})$ under $\mathbb{P}_x$. In particular, this yields a strictly positive lower bound on $\mathbb{P}_x(\tau_U < \infty)$. 

\captionsetup[algorithm]{name=Procedure}

\begin{algorithm}[h]
\caption{Coverability run construction for an effective pWSTS}
\label{alg:witness-path}
\begin{algorithmic}[1]
\Require $x \in S$, $\mathrm{In}_U$, $\mathrm{succ}$, $P$
\Ensure $\emptyset$ and $0$ if $x \notin \Pred^*(U)$; otherwise a finite run
$y_0 \cdots y_k$ with $y_0 = x$ and $y_k \in U$, along with its probability
$p^*$
\If{$\mathrm{In}_U(x,-1) = 0$} \Return $\emptyset, 0$ \EndIf
\State $k \gets 0$
\While{$\mathrm{In}_U(x, k) = 0$} \State $k \gets k + 1$ \EndWhile
\If{$k = 0$} \Return $x, 1$ \EndIf
\State $y_0 \gets x$; \; $p^* \gets 1$
\For{$i \gets 0$ \textbf{to} $k - 1$}
  \State $m \gets 1$; \; $z \gets \mathrm{succ}(y_i,m)$
  \While{$\mathrm{In}_U(z, k - i - 1) = 0$}
    \State $m \gets m + 1$; \; $z \gets \mathrm{succ}(y_i,m)$
  \EndWhile
  \State $y_{i+1} \gets z$; \; $p^* \gets p^* \times P(y_i)(y_{i+1})$
\EndFor
\State \Return $y_0\, y_1 \cdots y_k,\ p^*$
\end{algorithmic}
\end{algorithm}

\begin{proposition}
Procedure~\ref{alg:witness-path} terminates, and it outputs a valid minimum-length run
together with its probability if and only if $x \in \Pred^*(U)$.
\end{proposition}

\begin{proof}
If $x \notin \Pred^*(U)$, the procedure terminates and returns $\emptyset, 0$.
Otherwise, since $x \in \Pred^*(U) = \bigcup_{j \geq 0} \Pred^{\le j}(U)$,
there exists $j \in \Nat$ with $\mathrm{In}_U(x, j) = 1$, so the first loop
terminates; it exits at the least such $j$, denoted $k$. If $k = 0$, then
$x \in U$ and the run $x$ is returned with probability $1$. Assume $k \geq 1$.
Then $x \in \Pred^{\le k}(U) \setminus \Pred^{\le k-1}(U)$, hence
$x \in \Pred^{k}(U)$, and every run from $x$ to $U$ has length at least $k$.

For the outer loop, we prove the invariant $y_i \in \Pred^{k-i}(U)$ for
$i \leq k$, which holds at $i = 0$. Assume $y_i \in \Pred^{k-i}(U)$ for some
$i < k$. By definition of $\Pred$, there exists
$z \in \Post(y_i) \cap \Pred^{k-i-1}(U)$. This $z$ appears at some finite
index of the enumeration of $\Post(y_i)$, so all its calls to $\mathrm{succ}(y_i,\cdot)$
terminate and it sets $y_{i+1} \in \Post(y_i) \cap \Pred^{\le k-i-1}(U)$.
If we had $y_{i+1} \in \Pred^{j}(U)$ for some $j < k-i-1$, then, since
$y_0 \cdots y_{i+1}$ is a run from $x$, the state $x$ would reach $U$ in
$i+1+j \le k-1$ steps, contradicting the minimality of $k$. Hence
$y_{i+1} \in \Pred^{k-i-1}(U)$, and the invariant is preserved.

The outer loop performs exactly $k$ iterations, and at $i = k$ the invariant
gives $y_k \in \Pred^0(U) = U$. By the same minimality argument, $y_i \notin U$
for $i < k$, and since every run from $x$ to $U$ has length at least $k$, the
returned run has minimum length. Moreover, at termination,
$$
  p^\ast = \prod_{i=0}^{k-1} P(y_i)(y_{i+1})
  = \mathbb{P}_x\Big(\bigcap_{0 \leq i \leq k} \{X_i = y_i\}\Big) > 0,
$$
the last equality holding since $y_0 = x$.
\end{proof}

Note that the returned probability is a computable real: it is a finite
product of values $P(y_i)(y_{i+1})$. This algorithm addresses the computability of a shortest run, given the membership of the initial state in $\Pred^*(U)$. Here, the backward computation provides what the forward one lacks. The
intermediate bases guarantee that a useful successor exists, which $\Post$ does not provide, and this is what rules out the non-termination a blind interleaved search would face. To the best of our knowledge, it has not been formulated in this way before. In the sequel, this algorithm is useful as it returns a strictly positive lower bound on the actual
probability, which the other algorithms do not guarantee.

\section{Path-Enumeration Algorithms}

\label{sec:Path-Enumeration}

The approximate quantitative coverability problem (AQCP) that we tackle in this paper is formally defined as follows.

\noindent\textbf{Input:} an effective pWSTS $(\mathcal{S}, P)$, a finite set $B$, an initial state $x \in S$, a horizon $\ell \in \Nat \cup \{\infty\}$, and a precision $\varepsilon \in \mathbb{Q} \cap (0,1)$.

\noindent\textbf{Output:} a value $p^\ast \in \mathbb{Q}$ such that:
 $ p^\ast \;\leq\; \mathbb{P}_x(\tau_{\uparrow B} < \ell + 1) \;\leq\; p^\ast + \varepsilon $

Standard path-enumeration algorithms \cite{IyerNarasimha1997,AbdullaBenHendaMayr2007} typically approach this problem using breadth-first traversal, which requires the branching degree to be finite. In contrast, our framework imposes no such assumption. 

\begin{remark}[Precision of the transition probabilities]
\label{rem:precision}
For the sake of simplicity, we denote the transition probabilities by $P(x)(y)$ in the algorithms. In practice, we consider rational lower-bounding approximations, computed with sufficient precision to guarantee termination (by Lemma~\ref{lem:approx-termination}). The non-impact of these approximations on correctness is detailed in the proofs of the theorems.
\end{remark}

\subsection{Path-enumeration With Bounded Horizon}

A difficulty is that, for an infinite-branching pWSTS, the number of
runs of length $\ell$ from $x$ might be infinite, so the exact sum $\sum_{\pi}
\mathbb{P}(\pi)$ over accepting runs $\pi$ cannot be computed term by term. The
procedure \ref{alg:bounded-horizon} circumvents this by pruning, at every node, the
low-probability successors, keeping the total discarded mass below
$\varepsilon$. We let $B'$ denote a finite basis of
$\Pred^*( {\uparrow}B)$, computed by the backward coverability algorithm. We define $\mathrm{In}$, for $x \in S$ and $F \in \{B, B'\}$,  as follows:
$$
  \mathrm{In}(x, F) = 
  \begin{cases} 
    1 & \text{if } x \in {\uparrow}F, \\ 
    0 & \text{otherwise.} 
  \end{cases}
$$


\begin{algorithm}[!htbp]
\caption{Bounded-horizon approximation}
\label{alg:bounded-horizon}
\begin{algorithmic}
\Require $B$, $B'$, $\mathrm{In}$, $x \in S$, $\ell \in \Nat$,
$\varepsilon \in \mathbb{Q} \cap (0,1)$, $\mathrm{succ}$, $P$
\Ensure $p^\ast \in \mathbb{Q}$ such that
$p^\ast \leq \mathbb{P}_x(\tau_{ {\uparrow}B} \leq \ell) \leq p^\ast + \varepsilon$
\If{$\ell = 0$} \Return $\mathrm{In}(x, B)$ \EndIf
\State $p^\ast \gets 0$; \; $\gamma \gets \varepsilon / \ell$
\State $F \gets \{\, x \mapsto 1 \,\}$
  \Comment{frontier at depth $0$: a finite map from states to probabilities}
\For{$k = 0$ \textbf{to} $\ell - 1$}
  \State $F' \gets \emptyset$
  \ForAll{$(y, p) \in F$}
    \If{$\mathrm{In}(y, B) = 1$} \State $p^\ast \gets p^\ast + p$
    \ElsIf{$\mathrm{In}(y, B') = 1$}
      \State $q \gets 0$; \; $m \gets 1$
      \State \textbf{while} $q < 1 - \gamma$ \textbf{do}
      \State \quad $z \gets \mathrm{succ}(y, m)$; \;
      $q \gets q + P(y)(z)$; \; $m \gets m + 1$
      \State \quad $F'[z] \gets F'[z] + p \cdot P(y)(z)$
        \Comment{merge; absent keys read as $0$}
      \State \textbf{end while}
    \EndIf
  \EndFor
  \State $F \gets F'$
\EndFor
\ForAll{$(y, p) \in F$}
  \Comment{last level: no further expansion}
  \If{$\mathrm{In}(y, B) = 1$} \State $p^\ast \gets p^\ast + p$ \EndIf
\EndFor
\State \Return $p^\ast$
\end{algorithmic}
\end{algorithm}

\begin{lemma}[Termination]
\label{lem:bounded-horizon-termination}
Procedure~\ref{alg:bounded-horizon} terminates.
\end{lemma}

\begin{proof}
If $\ell = 0$ the algorithm returns immediately, so assume $\ell \geq 1$.
The outer loop performs exactly $\ell$ iterations, so it suffices to show
that each of them terminates, which amounts to showing that every frontier
$F_k$ is finite. Note that $\mathrm{In}(y, B') = 1$ and $\mathrm{In}(y, B) = 0$ implies $\Post(y) \neq \emptyset$.

We proceed by induction on $k$. The frontier $F_0 = \{x \mapsto 1\}$ is a
singleton. Assume $F_k$ is finite and consider an entry $(y,p) \in F_k$. If
$\mathrm{In}(y,B) = 1$ or $\mathrm{In}(y,B') = 0$, the entry
contributes nothing to $F_{k+1}$. Otherwise, since $P(y) \in \Dist(S)$ has
partial sums converging to $1$ and $\gamma \in \mathbb{Q}\cap(0,1)$, there is a finite $m$ with $q > 1 - \gamma$ (for a fine enough precision, by Lemma~\ref{lem:approx-termination}). The inner loop thus performs finitely many
iterations and inserts finitely many keys into $F_{k+1}$; each iteration is
effective, as $\mathrm{succ}$ is computable and $P$ is computable.  The inner \textbf{for all} loop thus ranges over a finite set and
each of its iterations terminates, so $F_{k+1}$ is finite and computed in
finite time. The final pass over $F_\ell$ likewise terminates, $F_\ell$
being finite.

\end{proof}

Having established its termination, this procedure constitutes an algorithm. We now turn to the proof of its correctness.

\begin{theorem}
\label{thm:bounded-horizon-correctness}
Algorithm~\ref{alg:bounded-horizon}, combined with the backward coverability algorithm, effectively solves the AQCP for $\ell \in \Nat$.
\end{theorem}

\begin{proof}
Let $U = {\uparrow}B$. In this proof, $P(y)(z)$ denotes the exact transition probability, and $a_z$
the rational lower bound that the algorithm actually uses — written
$P(y)(z)$ in the pseudo-code, see Remark~\ref{rem:precision}. We must make this distinction here to ensure that the approximations do not compromise the correctness of the algorithm. It is trivial for $\ell = 0$, so we assume from now on that $\ell \in \Nat_+$. At each depth $k \le \ell$, we define the potential:
$ \Phi_k = p^\ast_k + \sum_{(y,p) \in F_k} p\, \mathbb{P}_y(\tau_U \le \ell - k) $
where initially $\Phi_0 = \mathbb{P}_x(\tau_U \le \ell)$. We have for $k \le \ell$, $p^\ast_{k} \le \Phi_k$.

When transitioning from depth $k$ to $k+1$, we have $\mathbb{P}_y(\tau_U \le \ell - k) = 1$ for $y \in U$, and $\mathbb{P}_y(\tau_U \le \ell - k) = 0$ for $y \in \widetilde{U}$. For the remaining expanded states $y \notin U \cup \widetilde{U}$, the Markov property yields:
$ \mathbb{P}_y(\tau_U \le \ell - k) = \sum_z P(y)(z)\, \mathbb{P}_z(\tau_U \le \ell - k - 1) $

Let $z_1^{(y)}, \dots, z_{M_{y,k}}^{(y)}$ be the successors enumerated by the inner loop for state $y$ at depth $k$, with rational lower bounds $a_{z_i^{(y)}} \le P(y)(z_i^{(y)})$ (obtained with sufficient precision to guarantee termination according to Lemma~\ref{lem:approx-termination}) and total computed mass $q_{M_{y,k}} = \sum_{i=1}^{M_{y,k}} a_{z_i^{(y)}}$. 

By definition of the algorithm,
$p^\ast_{k+1} = p^\ast_k + \sum_{(y,p) \in F_k,\, y \in U} p$.
To handle state merging, let
$
  \mathcal{T}_k = \{\, (y,p,i) \mid (y,p) \in F_k,\
  y \notin U \cup \widetilde{U},\ 1 \leq i \leq M_{y,k} \,\}
$
be the set of triples indexing the insertions performed at depth $k$, and
for $z \in S$ let $\mathcal{A}(z,k) = \{\, (y,p,i) \in \mathcal{T}_k \mid
z_i^{(y)} = z \,\}$, where $z_1^{(y)}, \dots, z_{M_{y,k}}^{(y)}$ are the
successors enumerated from $y$ at depth $k$. By definition of the merge, the keys of
$F_{k+1}$ are the states $z$ with $\mathcal{A}(z,k) \neq \emptyset$, and the
associated probability is
$p' = \sum_{(y,p,i) \in \mathcal{A}(z,k)} p\,a_{z_i^{(y)}}$.

The sets $\mathcal{A}(z,k)$ partition $\mathcal{T}_k$, and the factor
$\mathbb{P}_z(\tau_U \le \ell - k - 1)$ depends on the triple only through
$z = z_i^{(y)}$. Reindexing the sum over $\mathcal{T}_k$ by $z$ on the one
hand and by $(y,p)$ on the other therefore gives
\begin{align*}
  \sum_{(z,p') \in F_{k+1}} p'\, \mathbb{P}_z(\tau_U \le \ell - k - 1)
  &= \sum_{(y,p,i) \in \mathcal{T}_k}
     p\,a_{z_i^{(y)}}\, \mathbb{P}_{z_i^{(y)}}(\tau_U \le \ell - k - 1) \\
  &= \sum_{(y,p) \in F_k,\, y \notin U \cup \widetilde{U}} p
     \sum_{i=1}^{M_{y,k}} a_{z_i^{(y)}}\,
     \mathbb{P}_{z_i^{(y)}}(\tau_U \le \ell - k - 1) .
\end{align*}

By splitting the sum in $\Phi_k$ over the states $y \in U$, $y \in \widetilde{U}$, and the remaining states, we have $\Phi_k = p^\ast_{k+1} + \sum_{(y,p) \in F_k, y \notin U \cup \widetilde{U}} p\, \mathbb{P}_y(\tau_U \le \ell - k)$. Subtracting $\Phi_{k+1}$ directly yields:
$$ \Phi_k - \Phi_{k+1} = \sum_{(y,p) \in F_k, y \notin U \cup \widetilde{U}} p \left( \mathbb{P}_y(\tau_U \le \ell - k) - \sum_{i=1}^{M_{y,k}} a_{z_i^{(y)}} \mathbb{P}_{z_i^{(y)}}(\tau_U \le \ell - k - 1) \right) $$

Since probabilities are bounded by 1 and the inner loop exit condition guarantees $q_{M_{y,k}} \ge 1 - \gamma$, the local error at state $y$ is bounded by:
$$ 0 \le \mathbb{P}_y(\tau_U \le \ell - k) - \sum_{i=1}^{M_{y,k}} a_{z_i^{(y)}} \mathbb{P}_{z_i^{(y)}}(\tau_U \le \ell - k - 1) \le 1 - q_{M_{y,k}} \le \gamma $$

Given that $\sum_{(y,p) \in F_k} p \le 1$, it follows from the bounds above that $0 \le \Phi_k - \Phi_{k+1} \le \gamma$. Hence, by telescoping over the $\ell$ steps,
$\Phi_0 - \Phi_\ell \leq \ell\gamma = \varepsilon$. Since
$\mathbb{P}_y(\tau_U \leq 0) = \mathbf{1}_U(y)$, the final pass over
$F_\ell$ gives $\Phi_\ell = p^\ast$. We conclude
$ p^* = \Phi_\ell \le \Phi_0 \le \Phi_\ell + \ell\gamma \le p^* + \varepsilon $, with $p^\ast \in \mathbb{Q}$.
\end{proof}

Algorithm~\ref{alg:bounded-horizon} computes the probability of reaching ${\uparrow}B$
within $\ell$ steps to any prescribed precision, without any decisiveness
assumption. This is one of its strengths: it
makes the algorithm applicable to every effective pWSTS. Notice that a similar algorithm computes a value $p^\ast \in \mathbb{Q}$ such that
$
  p^\ast \;\leq\; \mathbb{P}_x(\tau_{ {\uparrow}B\cup \widetilde{{\uparrow}B}} \leq \ell)
  \;\leq\; p^\ast + \varepsilon .
$
It suffices to replace $\mathrm{In}(y,B) = 1$ by
$\mathrm{In}(y,B) = 1 $ or $\mathrm{In}(y,B') = 0$, so that the
states of $\widetilde{{\uparrow}B}$ contribute to $p^\ast$ instead of being pruned.
Correctness is unaffected and the proofs of apply verbatim, neither of them using
the upward-closedness of the target.

Note that this algorithm also applies to the finitely branching case, where
for $\varepsilon$ small enough, it discards no probability mass.

\begin{proposition}
\label{prop:calcul-branchement-fini}
Assume for every $x \in S$, $\Post(x)$ is finite. Then, for $x \in S$ and $\ell \in \Nat$, there exists
$\epsilon_0 \in \mathbb{Q} \cap (0,1)$ such that for every
$\varepsilon \in \mathbb{Q} \cap (0, \epsilon_0)$,
Algorithm~\ref{alg:bounded-horizon} discards no probability mass during the computation steps.
\end{proposition}
\begin{proof}
If $\ell = 0$, the algorithm returns immediately, so the claim is immediate. Assume henceforth that $\ell \geq 1$. Since the model is finitely branching, $\Post^{\leq \ell}(x)$ is finite for
any $x \in S$, and so is the set of transitions
$\{P(y)(z): y \in \Post^{\leq \ell}(x),\ z \in \Post(y)\}$. All these
probabilities are strictly positive, so their minimum
$\epsilon_0:= \min\{P(y)(z): y \in \Post^{\leq \ell}(x),\ z \in
\Post(y)\}$ exists and satisfies $\epsilon_0 > 0$. Fix $\varepsilon \in
\mathbb{Q} \cap (0, \epsilon_0)$, so we have $\gamma = \varepsilon / \ell 
<\epsilon_0$. We claim that every inner loop of
Algorithm~\ref{alg:bounded-horizon} exits after
enumerating the whole (finite) support of $P(y)$. Suppose not: some state
$y \in \Post^{\leq \ell}(x)$ has its inner loop exit with $q_y > 1 - \gamma$
and at least one successor $z_m$ remains, with
$q_y + P(y)(z_m) \leq 1$, hence $q_y \leq 1 - P(y)(z_m)$. Combining with
$q_y > 1 - \gamma$ gives $1 - \gamma < 1 - P(y)(z_m)$, that is
$P(y)(z_m) < \gamma < \epsilon_0$. But $y \in \Post^{\leq \ell}(x)$ and
$z_m \in \Post(y)$, so $P(y)(z_m) \geq \epsilon_0$ by definition of
$\epsilon_0$, we have a contradiction. Every inner loop therefore enumerates the
entire support of $P(y)$, so no probability mass is ever discarded: the
algorithm sums the probabilities of all runs of length at most $\ell$
reaching $U$, and returns $\mathbb{P}_x(\tau_U \leq \ell)$.
\end{proof}

\subsection{Path-Enumeration With Infinite Horizon}

The principle is the same as the previous procedure, with one difference:
since the exploration depth is now unbounded, the truncation budget cannot
be constant. We take a sequence $(\gamma_k)$ summing to $\varepsilon/2$, so
that the total discarded mass stays below that value, and stop when the mass
remaining in the frontier drops below $\varepsilon/2$ as well. Decisiveness
guarantees that it eventually does, hence that the exploration depth is
finite.
As before, based on the backward algorithm, we define a computable function In such that for $x \in S, F \in \{B, B'\}$: $$
  \mathrm{In}(x, F) =
  \begin{cases}
    1 & \text{if } x \in {\uparrow}F, \\
    0 & \text{otherwise.}
  \end{cases}
$$

\begin{algorithm}[!htbp]
\caption{Infinite horizon approximation}
\label{alg:dynamic-truncation}
\begin{algorithmic}[1]
\Require $B$, $B'$, $\mathrm{In}$, $x \in S$,
$\varepsilon \in \mathbb{Q} \cap (0,1)$, $\mathrm{succ}$, $P$  
\Ensure $p^\ast \in \mathbb{Q}$ such that
$p^\ast \leq \mathbb{P}_x(\tau_{\uparrow B} < \infty) \leq p^\ast + \varepsilon$
\State $p^\ast \gets 0$; \; $k \gets 0$
\State $F \gets \{\, x \mapsto 1 \,\}$
  \Comment{frontier at depth $k$: a finite map from states to probabilities}
\While{$\sum_{(y,p) \in F} p > \varepsilon/2$}
  \State $\gamma_k \gets \varepsilon / 2^{\,k+2}$
    \Comment{$\sum_{k\ge0} \varepsilon\,2^{-(k+2)} = \varepsilon/2$:
    accumulated branching error bounded by $\varepsilon/2$}
  \State $F' \gets \emptyset$
  \ForAll{$(y, p) \in F$}
    \If{$\mathrm{In}(y, B) = 1$} \State $p^\ast \gets p^\ast + p$
    \ElsIf{$\mathrm{In}(y, B') = 1$}
      \State $q \gets 0$; \; $m \gets 1$
      \State \textbf{while} $q < 1 - \gamma_k$ \textbf{do}
      \State \quad $z \gets \mathrm{succ}(y, m)$; \;
      $q \gets q + P(y)(z)$; \; $m \gets m + 1$
      \State \quad $F'[z] \gets F'[z] + p \cdot P(y)(z)$
        \Comment{merge; absent keys read as $0$}
      \State \textbf{end while}
    \EndIf
  \EndFor
  \State $k \gets k + 1$
  \State $F \gets F'$
\EndWhile
\State \Return $p^\ast$
\end{algorithmic}
\end{algorithm}

We denote by $U$ the set $\uparrow B$.


\begin{proposition}[Termination]
\label{lem:dynamic-termination}
If the Markov Chain is decisive with respect to U, then Procedure~\ref{alg:dynamic-truncation} terminates.
\end{proposition}

\begin{proof}
We first check that every iteration of the main loop terminates, which
amounts to showing that every frontier $F_k$ is finite and 
computable. We proceed by
induction on $k$. The frontier $F_0 = \{x \mapsto 1\}$ is a singleton.
Assume $F_k$ is finite: an entry $(y,p) \in F_k$ with $\mathrm{In}(y,B) = 1$
or $\mathrm{In}(y,B') = 0$ contributes nothing to $F_{k+1}$, and
otherwise, by Lemma~\ref{lem:approx-termination} (since the state cannot be a deadlock one), its inner loop terminates
and inserts finitely many keys into $F_{k+1}$. The \textbf{for all} loop
thus ranges over a finite set and each of its iterations terminates, so
$F_{k+1}$ is finite and computed in finite time.

It remains to show that the main loop is left after finitely many
iterations. The chain is decisive with respect to $U$, so
$\mathbb{P}_x(\tau_{U \cup \widetilde{U}} < \infty) = 1$ therefore
$\lim _{n\to \infty}\mathbb{P}_x(\tau_{U \cup \widetilde{U}} > n) = 0$. Let $N \in \Nat$ be
such that $\mathbb{P}_x(\tau_{U \cup \widetilde{U}} > N) \leq
\varepsilon/2$. The entries of $F_{N+1}$ arise from runs $\pi = x_0...x_{N+1}$ along which
$x = x_0 \dots, x_N \notin U \cup \widetilde{U}$, and, summing the
probabilities of the runs reaching the same state at the same depth,
distinct entries correspond to disjoint events, so their accumulated
probabilities satisfy
$\sum_{(y,p) \in F_{N+1}} p \leq \mathbb{P}_x(\tau_{U \cup \widetilde{U}} >
N) \leq \varepsilon/2$. The condition of the main loop then fails at depth
$N+1$.
\end{proof}

Note that without the assumption of decisiveness, this procedure might not terminate. Consequently, in what follows, we treat it as a semi-algorithm. We now turn to establishing its correctness.

\begin{theorem}
\label{thm:dynamic-correctness}
If the effective pWSTS is decisive with respect to ${\uparrow}B$, Semi-algorithm~\ref{alg:dynamic-truncation}, combined with the backward coverability algorithm, effectively solves the AQCP for $\ell = \infty$.
\end{theorem}

\begin{proof}
For $y \in S$, we write $V(y) = \mathbb{P}_y(\tau_U < \infty)$.  As before, $P(y)(z)$ denotes the exact transition probability, and $a_z$
the rational lower bound that the algorithm actually uses, see Remark~\ref{rem:precision}. 

At each iteration $k \in \Nat$, let $F_k$ denote the frontier and $p^\ast_k$ the value of $p^\ast$ at the start of the iteration. We define the potential:
$ \Phi_k = p^\ast_k + \sum_{(y,p) \in F_k} p\, V(y) $
with initially $\Phi_0 = V(x) = \mathbb{P}_x(\tau_U < \infty)$.

During iteration $k$, consider a node $(y,p) \in F_k$. If $y \in U$, then $V(y) = 1$ and $p$ is added to $p^\ast$. If $y \in \widetilde{U}$, then $V(y) = 0$ and the node is discarded. Otherwise, $y \notin U \cup \widetilde{U}$ is expanded into enumerated successors $z_1, \dots, z_{M_y}$ with rational lower bounds  $0 \le a_{z_i} \le P(y)(z_i)$. As in the proof of Theorem~\ref{thm:bounded-horizon-correctness}, these lower bounds are computed with sufficient precision to guarantee termination (according to Lemma~\ref{lem:approx-termination}). The algorithm ensures that the evaluated probability mass $q_{M_y} = \sum_{i=1}^{M_y} a_{z_i}$ satisfies $q_{M_y} \ge 1 - \gamma_k$.

Following the exact same reasoning as in the proof of Theorem~\ref{thm:bounded-horizon-correctness}:
$$ \Phi_k - \Phi_{k+1} = \sum_{(y,p) \in F_k, y \notin U \cup \widetilde{U}} p \left( V(y) - \sum_{i=1}^{M_y} a_{z_i} V(z_i) \right) $$

By the Markov property, $V(y) = \sum_z P(y)(z)\, V(z)$ for any $y \notin U$. Since $V(z) \le 1$, the local error at state $y$ is bounded by:
$0 \le V(y) - \sum_{i=1}^{M_y} a_{z_i} V(z_i) \le 1 - q_{M_y} \le \gamma_k .$ Because the accumulated probabilities on the frontier always satisfy $\sum_{(y,p) \in F_k} p \le 1$, summing these local errors over $F_k$ bounds the total loss of potential per iteration:
$ 0 \le \Phi_k - \Phi_{k+1} \le \gamma_k $.

Let $K$ be the iteration at which the main loop exits. By telescoping the sequence $(\Phi_k)$, we obtain:
$$ p^\ast_K \le \Phi_K \le \Phi_0 \le \Phi_K + \sum_{k < K} \gamma_k \le p^\ast_K + \sum_{(y,p) \in F_K} p + \frac{\varepsilon}{2} \le p^\ast_K + \varepsilon $$
where the final two inequalities directly rely on $\sum_k\gamma_k \le \varepsilon / 2$ and the exit condition of the main loop ($\sum_{(y,p) \in F_K} p \le \frac{\varepsilon}{2}$). Moreover, $p^\ast_K \in \mathbb{Q}$ and according to Proposition~\ref{lem:dynamic-termination}, the procedure terminates under the decisiveness assumption, which yields the result.
\end{proof}

This semi-algorithm solves the approximate quantitative coverability problem over an infinite horizon, provided the upward closure of the target set is decisive. This holds regardless of the branching degree, finite or infinite, as soon
as the effectiveness assumptions and decisiveness are satisfied, and the
only probabilities the model has to compute are those of individual
transitions.

From these algorithms, we draw three remarks.

\begin{remark}[Merging states]
An optimization used in both approximation procedures is state merging. Following \cite{BarbotBouyerHaddad2024}, our algorithms merge distinct runs reaching the same state at the same depth by summing their probabilities. Standard breadth-first path-enumeration algorithms maintain one entry per path \cite{IyerNarasimha1997,AbdullaBenHendaMayr2007}, leading to an exponential or even factorial growth of the frontier (depending on $P$). Merging mitigates this combinatorial explosion by bounding the frontier size by the number of \emph{distinct} states reached at depth $k$.
\end{remark}

\begin{remark}[Choice of the budget allocation]
\label{rem:budget-allocation}
In Semi-algorithm~\ref{alg:dynamic-truncation}, the sequence $(\gamma_k)$ of
per-level budgets can be any sequence of positive rationals with
$\sum_{k \geq 0} \gamma_k \leq \varepsilon/2$; the geometric choice
$\gamma_k = \varepsilon\,2^{-k-2}$ is only one possibility. The speed at
which $(\gamma_k)$ decreases has a concrete effect, because a smaller
$\gamma_k$ forces the inner loop at depth $k$ to enumerate more successors
before the residual mass falls below it. A fast-decreasing allocation, such as the geometric one, leaves a
comfortable budget at shallow levels but a very small one deep down, where
$\gamma_k$ becomes so small that distinguishing $1 - \gamma_k$ from $1$
requires many bits of precision. A slow-decreasing allocation, such as
$\gamma_k \propto 1/k^2$, keeps $\gamma_k$ larger at depth and avoids this,
but since the total budget is fixed, it must then be tighter at the shallow
levels, which are precisely the ones carrying most of the probability mass.
No allocation escapes this trade-off: a fixed budget spread over infinitely
many levels is always small somewhere.
\end{remark}

\begin{remark}[The missing a-priori bound]
\label{rem:gamma-comparison}
Decisiveness guarantees that Semi-algorithm~\ref{alg:dynamic-truncation}
terminates, but says nothing about the depth at which it does: no integer
$N$ with $\mathbb{P}_x(\tau_{U \cup \widetilde{U}} > N) \leq \varepsilon/2$
is known in general before the algorithm is run. The finest precision to which the transition probabilities have to be
computed is therefore unknown as well, since it is governed by $\gamma_k$. In particular, even when the branching
is finite, neither the time nor the space used by the algorithm can be
predicted in advance, both depending on how deep the exploration goes.
Knowing such an $N$ beforehand would remove this limitation. This
investigation is left for future work.
\end{remark}

Note also that neither Algorithm~\ref{alg:bounded-horizon} nor
Semi-algorithm~\ref{alg:dynamic-truncation} guarantees $p^\ast > 0$. If the
probability to be estimated is $10^{-2}$ and one takes
$\varepsilon = 10^{-1}$, then $p^\ast = 0$ is admissible. Algorithm~\ref{alg:witness-path} fills this gap. Indeed, obtaining a strictly positive lower bound is much more straightforward with this approach: since it isolates a specific run, we have direct access to its computable real probability. By approximating this single probability with sufficient precision, the algorithm readily returns a strictly positive lower bound $p_0$. In contrast, ensuring $p^\ast > 0$ with the other procedures would require repeatedly restarting them with smaller values of $\varepsilon$. Running the latter with
$\varepsilon = \alpha p_0$, for $\alpha \in \mathbb{Q} \cap (0,1)$, then
bounds the error by a fraction $\alpha$ of the actual probability, since
$\alpha p_0 \leq \alpha\,\mathbb{P}_x(\tau_U < \infty)$.

\section{Application}
\label{sec:applications}

We instantiate our framework on multi-type Galton--Watson processes, a
natively probabilistic model. It does not arise from a WSTS later equipped
with probabilities, but is defined as a stochastic process whose underlying
transition system turns out to be a WSTS. We show that these processes are
pWSTS, that they are effective under several assumptions on their
reproduction laws, and that they are decisive with respect to every
upward-closed set.

\subsection{Galton--Watson Process as a pWSTS}

\begin{example}[Two-strain epidemic with a capacity threshold]
\label{ex:mtgw-strains}
Consider a pathogen circulating in two forms: type~$1$ denotes hosts
carrying the wild strain, and type~$2$ hosts carrying a resistant strain.
Each infected host of type~$i$ (row) independently infects $\mathcal{P}(m_{ij})$
(a Poisson distribution with parameter $m_{ij} \in \mathbb{Q}$) hosts of type~$j$ (column) in the
next generation, with mean matrix
$$
  M = \begin{pmatrix} 0.4 & 0.2 \\ 0.6 & 0.8 \end{pmatrix} ,
$$
mutation at transmission and reversion of the resistant strain making both
types feed each other. Resistant cases, unlike wild-strain ones, require hospital care, of which
$N \geq 1$ beds are available. An outbreak is deemed severe when capacity is
overrun, that is when the population enters
  $U = \{ z \in \Nat^2: z_2 \geq N \}$,
an upward-closed set. Starting from a single wild-strain case $x = (1,0)$, the quantity of
interest is $\mathbb{P}_x(\tau_U < \infty)$. Extinction
at the very first generation already occurs with probability
$e^{-0.4} \times e^{-0.2} = e^{-0.6}$, so
$\mathbb{P}_x(\tau_U < \infty) \leq 1 - e^{-0.6} < 1$, while $U$ is
reachable from $x$ with positive probability.

By stability of the Poisson law under summation, the kernel is available:
$
  Z_1 \mid Z_0 = z \;\sim\;
  \mathcal{P}(0.4\,z_1 + 0.6\,z_2) \otimes \mathcal{P}(0.2\,z_1 + 0.8\,z_2) ,
$
so $\Post(z)$ is infinite for every $z \neq 0$ and standard
path-enumeration algorithms
\cite{IyerNarasimha1997,AbdullaBenHendaMayr2007} do not apply. 
\end{example}

This example is formally modeled as a multi-type Galton--Watson process. 

\begin{definition}
Let $d \in \Nat_+$ and for each $i \in \{1,...,d\}$, let $\mu^{(i)}$ be a probability distribution on $\Nat^d$. Let $(Z_n)_{n \in \Nat}$ in $\Nat^d$ be a stochastic process. $(Z_n)_{n \in \Nat}$
is a \emph{multi-type Galton--Watson process} associated with the family of laws of reproduction $\mu = (\mu^{(i)})_{1 \leq i \leq d}$ if:
$$
\begin{cases}
Z_0 = z_0 \in \Nat^d, \\[4pt]
Z_{n+1}^{(j)} = \displaystyle\sum_{i=1}^{d} \left( \sum_{k=1}^{Z_n^{(i)}} Y_{n,k}^{(i,j)} \right) \quad \text{for all } j \in \{1,..,d\} \text{ and } n \in \Nat,
\end{cases}
$$
where $(Y_{n,k}^{(i)})_{(n,k,i) \in \Nat\times\Nat_+ \times \{1,..,d\}}$ is a family of independent random vectors such that $Y_{n,k}^{(i)} = (Y_{n,k}^{(i,1)}, \ldots, Y_{n,k}^{(i,d)})$ is distributed according to $\mu^{(i)}$ for all $n\in \Nat$, $k \in \Nat_+$ and $i \in \{1,..,d\}$.

\end{definition}

The vector $Y_{n,k}^{(i)}$ represents the number of descendants of the $k$-th individual of type $i$ in the $n$-th generation, where each component $Y_{n,k}^{(i,j)}$ denotes the number of offspring of type $j$; the vector $Z_n = (Z_n^{(1)}, \ldots, Z_n^{(d)})$ represents the composition of the population at generation $n$, where $Z_n^{(j)}$ is the total number of individuals of type $j$.

\begin{remark}
Multi-type Galton--Watson processes follow synchronous semantics, every
individual reproducing at each generation, whereas stochastic models of
concurrent systems usually rely on asynchronous scheduling. A prominent
example is the probabilistic Basic Parallel Process (pBPP) model, where
exactly one process is selected at each step \cite{BonnetKieferLin2014}.
The two models differ on two points. First, the induced Markov chain is not that of a Galton--Watson process. Second, a pBPP is finitely branching, having at most as many
successors as it has rules, whereas a Galton--Watson process is infinitely
branching as soon as one reproduction law has infinite support.
Finally, Bonnet et al. \cite{BonnetKieferLin2014} establish the decidability of almost-sure coverability for
pBPPs, which is a qualitative question. We address the quantitative one,
giving effectiveness conditions under which the probability itself can be
approximated.
\end{remark}

The transition system is defined by $S = \Nat^d$ and, for $x, y \in \Nat^d$,
$x \to y$ iff $\mathbb{P}_x(Z_1 = y) > 0$. We equip $S$ with the
componentwise order.

\begin{restatable}{proposition}{gwPwsts}
\label{prop:gw-pwsts}
The system $(S, \leq, \to, P)$ associated with a multi-type Galton--Watson
process is a pWSTS.
\end{restatable}

\begin{proof}
    The transition system is defined by $S = \Nat^d$ and, for $x, y \in \Nat^d$,
$x \to y$ iff $\mathbb{P}_x(Z_1 = y) > 0$. We equip $S$ with the
componentwise order, which is a wqo by Dickson's lemma. Strong monotonicity holds. Indeed, let $x,x',y \in \Nat^d$ and assume $x \leq x'$ in $\Nat^d$ and $x \to y$, so $P(x)(y) > 0$. There exist vectors
$v^{(i)}_k \in \mathrm{supp}(\mu^{(i)})$, for $1 \leq i \leq d$ and
$1 \leq k \leq x_i$, such that
$y = \sum_{i=1}^{d} \sum_{k=1}^{x_i} v^{(i)}_k$. For each type $i$, fix an
arbitrary $\underline{v}^{(i)} \in \mathrm{supp}(\mu^{(i)})$ and set
$
  y' = y + \sum_{i=1}^{d} (x'_i - x_i)\,\underline{v}^{(i)} .
$
Assigning the offspring vectors $v^{(i)}_k$ to the individuals already
present in $x$ and $\underline{v}^{(i)}$ to the $x'_i - x_i$ additional
individuals of type $i$ is one way of producing $y'$ from $x'$, hence
$
  P(x')(y') \;\geq\;
  \Big(\prod_{i=1}^{d} \prod_{k=1}^{x_i} \mu^{(i)}\big(v^{(i)}_k\big)\Big)
  \cdot \prod_{i=1}^{d}
  \mu^{(i)}\big(\underline{v}^{(i)}\big)^{\,x'_i - x_i} \;>\; 0 ,
$
and $y' \geq y$ componentwise since each $\underline{v}^{(i)}$ has
non-negative entries. Thus $(S, \leq, \to)$ is a WSTS and, $\to$ being defined from $P$,
$(S, \leq, \to, P)$ is a pWSTS.
\end{proof}

\subsection{Effectiveness}

Under suitable assumptions on its reproduction laws, a multi-type Galton--Watson process is also an effective pWSTS. 

\begin{theorem} 
\label{theorem:Galton--Watson-effective}
    Let $(Z_n)_{n \in \Nat}$ be a multi-type Galton--Watson process  with the family of laws of reproduction $\mu = (\mu^{(i)})_{1 \leq i \leq d}$. Assume that for any $1 \leq i \leq d$:
    \begin{itemize}
        \item $\mu^{(i)}$ is computable 
        \item membership in $\mathrm{supp}(\mu^{(i)})$ is decidable
        \item determining for any $v$ whether there exists $v' \in \mathrm{supp}(\mu^{(i)})$ such that $v' \ge v$ is decidable.
    \end{itemize} 
    Then, the underlying pWSTS $(S, \le, \to, P)$ is effective. 
\end{theorem}

\begin{proof}

    \begin{itemize}
        \item \textbf{Decidable ordering and equality:} the state set is $S = \mathbb{N}^d$. The componentwise order $\le$ and equality $=$ on vectors of integers are decidable.
        
        \item \textbf{Pred-Basis effectivity:} for any target $b \in \mathbb{N}^d$, we show that a finite basis for $\Pred(\uparrow b)$ is computable. If $b = 0_{\Nat^d}$, then $\uparrow b = \Nat^d$, so $\Pred(\uparrow b) = \Nat^d$ (by monotonicity of $\to$, since $0 \to 0$). Thus, we output in this case the finite basis: $\{0_{\Nat^d}\}$. Otherwise, we proceed in three steps:

        \textbf{\textit{1. Combinatorial bound on minimal predecessors.}}
        Let $x \in \Pred(\uparrow b)$ (assuming $\Pred(\uparrow b) \neq \emptyset$) be a minimal predecessor. By definition, there exists a realization of reproduction vectors $v'_k{}^{(i)} \in \mathrm{supp}(\mu^{(i)})$ (for $1 \le i \le d$ and $1 \le k \le x_i$) such that $\sum_{i=1}^{d} \sum_{k=1}^{x_i} v'_k{}^{(i)} \ge b$. 
        We can extract contributions $0_{\Nat^d}\le c_k^{(i)} \le v'_k{}^{(i)}$ such that the sum exactly covers $b$:
        $$ \sum_{i=1}^{d} \sum_{k=1}^{x_i} c_k^{(i)} = b $$
        Let $\mathcal{K} = \{ (i,k) \mid c_k^{(i)} \neq 0_{\Nat^d} \}$ be the set of strictly useful terms. Note that $\mathcal{K} \neq \emptyset$ because $b \neq 0_{\Nat^d}$. Applying the $\ell_1$-norm to the equation yields $\sum_{(i,k) \in \mathcal{K}} \Vert c_k^{(i)} \Vert_1 = \Vert b \Vert_1$. Since $c_k^{(i)} \neq 0_{\Nat^d}$ implies $\Vert c_k^{(i)} \Vert_1 \ge 1$, we obtain $|\mathcal{K}| \le \Vert b \Vert_1$.
        
        Let $x' \in \mathbb{N}^d$ be the configuration restricted to $\mathcal{K}$, defined for each type $i$ by $x'_i = |\{ k \mid (i,k) \in \mathcal{K} \}|$. By construction, $x' \le x$ and $\Vert x' \Vert_1 = |\mathcal{K}| \le \Vert b \Vert_1$. Furthermore, the individuals retained in $x'$ suffice to cover $b$: assigning
to them the same vectors $v'^{(i)}_k$ as in $x$, restricted to the indices
$(i,k) \in \mathcal{K}$, yields
$$
  \sum_{(i,k) \in \mathcal{K}} v'^{(i)}_k
  \;\geq\; \sum_{(i,k) \in \mathcal{K}} c^{(i)}_k
  \;=\; \sum_{i=1}^{d} \sum_{k=1}^{x_i} c^{(i)}_k \;=\; b ,
$$
the second equality holding since $c^{(i)}_k = 0_{\Nat^d}$ outside
$\mathcal{K}$. This realization has positive probability, so
$x' \in \Pred(\uparrow b)$. Since $x$ is assumed to be a minimal predecessor, $x' \le x$ implies $x = x'$. Thus, any minimal predecessor is contained in the finite set $\mathcal{X}_b = \{ x \in \mathbb{N}^d \mid \Vert x \Vert_1 \le \Vert b \Vert_1 \}$. However, some states in $\mathcal{X}_b$ may not belong to $\Pred(\uparrow b)$. We therefore need to establish a sufficient and necessary condition for membership in $\Pred(\uparrow {b}).$

            \textbf{\textit{2. Finite space of useful contributions.}}
We define the finite set of contributions that an individual of type $i$ can ensure for the target $b$:
$$ V_i(b) = \{ v \in \mathbb{N}^d \setminus \{0_{\Nat^d}\} \mid v \le b \text{ and } \exists v' \in \mathrm{supp}(\mu^{(i)}), v' \ge v \} $$
Since $v$ is bounded by $b$ and by the third assumption, $V_i(b)$ is finite and computable. This provides a sufficient and necessary condition for membership in $\Pred(\uparrow b)$: if for $x \in \Nat^d$, there exists a combination of contributions $v_k^{(i)} \in V_i(b) \cup \{0_{\Nat^d}\}$ (for $1 \le i \le d$ and $1 \le k \le x_i$) such that $\sum_{i=1}^{d} \sum_{k=1}^{x_i} v_k^{(i)} \ge b$, then the probability of reaching $\uparrow b$ from $x$ is strictly positive, meaning $x \in \Pred(\uparrow b)$. Conversely, if $x \in \Pred(\uparrow b)$, then one can extract contributions $0_{\Nat^d} \le c_k^{(i)} \le v'_k{}^{(i)}$ such that the sum exactly covers $b$:
$ \sum_{i=1}^{d} \sum_{k=1}^{x_i} c_k^{(i)} = b $. For any strictly positive contribution, $c_k^{(i)} \in V_i(b)$. Hence, $c_k^{(i)} \in V_i(b) \cup \{0_{\Nat^d}\}$ for every $(i,k)$, and $\sum_{i,k} c_k^{(i)} = b \ge b$, so the condition holds. 
Note that the sets $V_i(b)$ may be empty. In particular, if $V_i(b) = \emptyset$ for all $1 \le i \le d$, since $b \neq 0_{\Nat^d}$, it yields $\Pred(\uparrow b) = \emptyset$.

\textbf{\textit{3. Computation of $\mathrm{cpre}(b)$.}}
We initialise $\mathcal{P}_b = \emptyset$ and, for each candidate
$x \in \mathcal{X}_b$, test whether some family
$(v^{(i)}_k)_{i,k} \in \prod_{1 \le i \le d} (V_i(b) \cup \{0_{\Nat^d}\})^{x_i}$
satisfies $\sum_{i,k} v^{(i)}_k \ge b$; if so, $x$ is added to
$\mathcal{P}_b$. Each such product is finite, so the test is decidable by
exhaustive enumeration, and $\mathcal{X}_b$ is finite and computable, thus the
whole procedure terminates. By Step~2, $\mathcal{P}_b =
\Pred(\uparrow b) \cap \mathcal{X}_b$, and by Step~1 every minimal
element of $\Pred(\uparrow b)$ lies in $\mathcal{X}_b$, hence in
$\mathcal{P}_b$. We get
$$ \mathrm{cpre}(b) = \min\nolimits_{\le}(\mathcal{P}_b)
   = \min\nolimits_{\le}(\Pred(\uparrow b)) , $$
a finite basis of the upward-closed set $\Pred(\uparrow b)$. In
particular $\mathcal{P}_b = \emptyset$ exactly when
$\Pred(\uparrow b) = \emptyset$, in which case the empty basis is
returned (since $b \neq 0_{\Nat^d}$).
 \item \textbf{Decidable $\to$:} As a corollary of the bounding logic applied in Step 2, the exact one-step transition relation $x \to y$ is decidable. By definition, a transition $x \to y$ occurs if and only if the configuration $y$ can be generated by the individuals in $x$, meaning $y = \sum_{i=1}^{d} \sum_{k=1}^{x_i} v'_k{}^{(i)}$ for some $v'_k{}^{(i)} \in \mathrm{supp}(\mu^{(i)})$. Since all reproduction vectors are non-negative, any vector contributing to this sum must satisfy $v'_k{}^{(i)} \le y$. Because membership in the supports is assumed to be decidable, one can compute the finite set of possible exact contributions for a target $y$:
$$ W_i(y) = \{ v \in \mathbb{N}^d \mid v \le y \text{ and } v \in \mathrm{supp}(\mu^{(i)}) \} $$
Deciding whether $x \to y$ reduces to checking if there exists an element of $\prod_{1 \le i \le d} (W_i(y))^{x_i}$ whose components sum to $y$. Since the sets $W_i(y)$ are finite, there is a finite number of combinations to enumerate, making the transition relation decidable.

            \item \textbf{Computable transition probabilities:} For any $x, y \in \mathbb{N}^d$, the transition probability $P(x)(y)$ is given by the discrete convolution of the reproduction laws $\mu^{(i)}$. Building on the logic of the previous step, this probability corresponds exactly to the sum of the probabilities of all valid combinations in $\prod_{1 \le i \le d} (W_i(y))^{x_i}$ whose components sum to $y$. The sum is ranging over families indexed by individuals, so that permutations of a same multiset of reproduction vectors are counted separately. If no such combination exists (i.e., $x \not\to y$), the probability is simply $0$. Otherwise, since the laws $\mu^{(i)}$ are assumed to be computable, $P(x)(y)$ evaluates to a finite sum of finite products of computable reals. Thus, P is computable.
            
        \item \textbf{Post effective enumerability: } Let $\mathsf{gen} \colon \mathbb{N} \to \mathbb{N}^d$ be a computable bijection enumerating the state set (e.g., by increasing sum of components $\Vert \cdot \Vert_1$). Since $\to$ is decidable, given $x$ and $m$, we can compute $\mathrm{succ}(x,m)$ by iterating $k = 0, 1, 2, \dots$, checking if $x \to \mathsf{gen}(k)$ and incrementing a counter $K$. This algorithm terminates and returns the $m$-th successor if $m \leq |\Post(x)|$, but does not terminate otherwise, proving $\Post$ is effectively enumerable. 
    \end{itemize}  
   
\end{proof}

Many distributions satisfy these three assumptions, including the
multidimensional Poisson distribution. Note that the laws need not be
identical across types: each $\mu^{(i)}$ may follow a different law, the
assumptions being required of each separately.

\subsection{Decisiveness}

We now prove that multi-type Galton--Watson processes are decisive with respect to every upward-closed set by showing that they are stochastically monotone.

\begin{lemma}
\label{lem:gw-monotone}
A pWSTS induced by a multi-type Galton--Watson process is stochastically monotone.
\end{lemma}

\begin{proof}
Let $x \leq x'$ in $\Nat^d$ and let $U \subseteq \Nat^d$ be upward-closed.
Let $(Y^{(i)}_k)_{1 \leq i \leq d,\, k \geq 1}$ be a family of independent
random vectors with $Y^{(i)}_k$ distributed according to $\mu^{(i)}$, and set
$$
  Z \;=\; \sum_{i=1}^{d} \sum_{k=1}^{x_i} Y^{(i)}_k ,
  \qquad
  Z' \;=\; \sum_{i=1}^{d} \sum_{k=1}^{x'_i} Y^{(i)}_k .
$$
By the definition of the process, $Z$ is distributed according to $P(x)$ and
$Z'$ according to $P(x')$. Since $x_i \leq x'_i$ for every $i$ and the
vectors $Y^{(i)}_k$ have non-negative entries, $Z'$ is obtained from $Z$ by
adding finitely many non-negative vectors, so $Z \leq Z'$ pointwise. As $U$
is upward-closed, $Z \in U$ implies $Z' \in U$, whence
$$
  P(x)(U) \;=\; \mathbb{P}(Z \in U)
  \;\leq\; \mathbb{P}(Z' \in U) \;=\; P(x')(U) . 
$$
\end{proof}

\begin{proposition}
    \label{cor:gw-decisive}
    Let $(Z_n)_{n \in \Nat}$ be a multi-type Galton--Watson process. For
every upward-closed $U \subseteq \Nat^d$, the process
is decisive with respect to $U$. 
\end{proposition}

\begin{proof}
    Immediate from Lemma~\ref{lem:gw-monotone} and Corollary~\ref{cor:pWSTS-decisive}.
\end{proof}

This proposition is worth noting for at least two reasons. First, it establishes a qualitative property of the Markov chain (decisiveness is a notion about Markov chains, not about WSTS), yet the proof goes entirely through the pWSTS formalism. Second, none of the usual tools for studying these processes is involved: the argument uses neither generating functions nor the regime of the process, and does not even require the reproduction laws to have finite means. Everything follows from viewing the process as a pWSTS, which illustrates the reach of our framework.

Returning to Example~\ref{ex:mtgw-strains}, the effectiveness requirements
are met by Theorem~\ref{theorem:Galton--Watson-effective}, and the process
is decisive with respect to $U = {\uparrow}(0,N)$ by
Proposition~\ref{cor:gw-decisive}. Thus,
Semi-algorithm~\ref{alg:dynamic-truncation} computes
$\mathbb{P}_x(\tau_U < \infty)$ up to any prescribed precision.

\begin{theorem}
\label{thm:gw-aqcp}
Let $(Z_n)_{n \in \Nat}$ be a multi-type Galton--Watson process satisfying
the effectiveness assumptions of Theorem~\ref{theorem:Galton--Watson-effective}. Then the
AQCP is effectively solved for every
$\ell \in \Nat \cup \{\infty\}$.
\end{theorem}

\section{Conclusion and open problems}
\label{sec:ccl}

We introduced probabilistic well-structured transition systems (pWSTS), a
framework for Markov chains whose underlying transition systems are WSTS. It
covers classical probabilistic models such as pVAS and pLCS, but also
infinitely branching systems and models that are probabilistic by
definition, such as multi-type Galton--Watson processes.

We then identified five natural effectiveness assumptions: decidability of
$=$ and of $\to$, Pred-Basis effectiveness, effective enumerability of
$\Post$, and computability of $P$. These relax the ones required for
effective Markov chains \cite{AbdullaBenHendaMayr2007}, where $\Post(x)$
must be computed explicitly, which confines the framework to finitely
branching models. Here, successors are enumerated only until the accumulated
probability mass bounds what is left unexplored, so effective enumerability
suffices and $\Post(x)$ may be infinite. The last assumption is also all
that is asked of the model in probabilistic terms: individual transition
probabilities, and not bounded-horizon ones, whose computation is far from
immediate and is deferred to the model in other approximation schemes
\cite{bertrand_et_al:LIPIcs.ICALP.2016.101}. These assumptions support three
procedures, none requiring knowledge of the branching degree or computable
bounded-horizon probabilities: Algorithm~\ref{alg:witness-path} computes a
minimum-length coverability run and its probability,
Algorithm~\ref{alg:bounded-horizon} solves approximate quantitative
coverability over a bounded horizon, and
Semi-algorithm~\ref{alg:dynamic-truncation} solves it over an infinite
horizon under decisiveness. Chaining the first as a lower bound $p_0 > 0$
with $\varepsilon = \alpha p_0$ in the others bounds the error by a fraction
$\alpha$ of the estimated quantity.

We also identified a general source of decisiveness: every stochastically
monotone pWSTS is decisive with respect to every upward-closed set. We
instantiated the framework on multi-type Galton--Watson processes, which are
pWSTS, effective as soon as their reproduction laws satisfy three
assumptions, and stochastically monotone. The approximate quantitative
coverability problem is thereby solved for these processes over both
horizons, which to the best of our knowledge had not been established
before. The proof uses none of the usual tools for branching processes: no
generating functions, no condition on the mean matrix, no moment assumption,
not even finiteness of the means, and no case distinction between the
subcritical, critical and supercritical regimes. The only work specific to
the model is to check the effectiveness assumptions, and they hold for many
usual reproduction laws, for example the multidimensional Poisson
distribution.

Several questions remain open. The first is to characterize the effective
decisive pWSTS for which, given an upward-closed target $U$ and an initial
state $x$, one can compute a decreasing function
$f_{U,x} \colon \mathbb{N} \to [0,1]$ converging to zero and satisfying
\[
  \mathbb{P}_x\!\left(\tau_{U \cup \widetilde U} > n\right)
  \leq f_{U,x}(n)
\]
for every $n$, together with an effectively computable threshold for each
precision $\varepsilon$. A second question is to determine when such bounds
can be made uniform over all initial states, or over broad classes of
targets. Answering these would clarify the boundary between qualitative
decisiveness and effective quantitative verification for probabilistic
infinite-state systems. A third direction is to develop a complexity
analysis of approximate quantitative coverability for effective pWSTS, and
to instantiate the resulting bounds on concrete subclasses.

\paragraph*{Use of AI.}
The authors take full responsibility for the entire content of this paper,
regardless of the assistance received from AI tools. The decision to address
the coverability problem by introducing a probabilistic class of WSTS is our
own, and was not suggested by an AI. The same holds for the choice of the
objects in their present form, and for the decision to include multi-type
Galton--Watson processes. The paper itself was written by the authors,
albeit with assistance.

We used large language models, namely Claude, Gemini, ChatGPT and
Perplexity, as follows. They helped explore various directions, some
fruitful and some not. The design of some examples, such as $S_1$ in
Example~\ref{ex:infinite-jumps}, and of some proofs, such as that of
Pred-Basis effectiveness (for Galton--Watson processes) and that of
Lemma~\ref{lem:monotone-reachability}, was AI-assisted. These tools also
helped gather a substantial body of related work on probabilistic
coverability, in particular for multi-type Galton--Watson processes, and
contributed to the \LaTeX{} writing of the statements and proofs. Some
proofs were reviewed by AI models, among them the proof of Pred-Basis
effectiveness for Galton--Watson processes. All proofs were checked by the
authors, AI review serving only as an additional verification. Finally, the
paper as a whole was proofread with AI assistance.

\bibliography{lipics-v2021-sample-article}

\end{document}